\documentclass[year=26]{fmcad}

\usepackage{cite}
\usepackage{amsmath,amssymb,amsfonts}
\usepackage{algorithmic}
\usepackage{graphicx}
\usepackage{textcomp}
\usepackage{xcolor}

\usepackage{url}
\usepackage{wrapfig}
\usepackage{lipsum}
\usepackage{pifont}
\usepackage[inline]{enumitem}
\usepackage[T1]{fontenc}
\usepackage{xcolor}
\usepackage[scaled=0.9]{beramono}
\usepackage{caption}
\usepackage{subcaption}
\usepackage{listings}
\usepackage{multirow, makecell}
\usepackage{booktabs}
\usepackage{amsthm}
\usepackage{tikz}
\usepackage{pgfplots}
\usepackage{xspace}

\hypersetup{colorlinks=true,breaklinks=true}

\theoremstyle{definition}
\newtheorem{theorem}{Theorem}
\newtheorem{definition}{Definition}
\newtheorem{lemma}{Lemma}

\newtheorem{fact}{Fact}

\usetikzlibrary{positioning}
\usetikzlibrary{tikzmark}
\usetikzlibrary{calc}
\tikzset{
    node1/.style={circle, draw=red!60, fill=red!5,minimum size=7mm},
    node2/.style={circle, draw=blue!60, fill=blue!5,minimum size=7mm},
    noderoot/.style={node2, very thick},
    subnet/.style={rectangle,rounded corners=.1cm,text=white,fill=red!70,minimum size=3mm,font=\small},
    annot/.style = {text width=4em, text centered, node distance = 6mm}
}

\pgfplotsset{width=10cm,compat=1.9}

\def\showcomments{0}

\ifnum\showcomments=1
\newcommand{\dz}[1]{{\footnotesize\color{blue}[dz: #1]}}
\newcommand{\dpw}[1]{{\footnotesize\color{blue}[dpw: #1]}}
\newcommand{\ag}[1]{{\footnotesize\color{purple}[ag: #1]}}
\newcommand{\tim}[1]{{\footnotesize\color{brown}[tim: #1]}}
\newcommand{\todo}[1]{{\footnotesize\color{red}[TODO: #1]}}
\else
\newcommand{\dz}[1]{}
\newcommand{\dpw}[1]{}
\newcommand{\ag}[1]{}
\newcommand{\tim}[1]{}
\newcommand{\todo}[1]{}
\fi

\newcommand{\OMIT}[1]{{}}

\newcommand{\IE}{i.e.,\xspace}
\newcommand{\EG}{e.g.,\xspace}

\newcommand{\sysname}{\textsc{CB-Ver}\xspace}

\newcommand{\synthname}{\textsc{CB-2IQ}\xspace}

\newcommand{\remove}[1]{} 

\newcommand{\cbgraph}{CB-graph\xspace}
\newcommand{\cbgraphs}{CB-graphs\xspace}
\newcommand{\cbedge}{CB-edge\xspace}
\newcommand{\cbedges}{CB-edges\xspace}
\newcommand{\cbroot}{CB-root\xspace}
\newcommand{\cbroots}{CB-roots\xspace}

\newcommand{\cbedgecheck}{CB-edge-check\xspace}

\newcommand{\convcheck}{\cbedgecheck}

\newcommand{\Invcheck}{Invariance\xspace}
\newcommand{\invcheck}{invariance\xspace}

\newcommand{\vc}[1]{VC_{\mathrm{#1}}}
\newcommand{\vcinit}{\vc{Init}}

\newcommand{\vcprop}{\vc{Prop}}

\newcommand{\vcinv}{\vc{Inv}}
\newcommand{\vcroot}{\vc{CBroot}}
\newcommand{\vcedge}{\vc{CBedge}}

\newcommand{\lp}{\mathrm{lp}}
\newcommand{\len}{\mathrm{len}}
\newcommand{\aspath}{\mathrm{aspath}}

\newcommand{\interface}{component specification\xspace}
\newcommand{\Interface}{Component specification\xspace}
\newcommand{\interfaces}{component specifications\xspace}

\newcommand{\Iinterface}{invariant specification\xspace}
\newcommand{\Qinterface}{abstract convergence specification\xspace}

\newcommand{\cbgenerate}{generate\xspace}
\newcommand{\cbgenerates}{generates\xspace}
\newcommand{\cbgenerated}{generated\xspace}

\newcommand{\cmark}{\ding{51}}
\newcommand{\xmark}{\ding{55}}

\newcommand{\bench}[1]{{\small\texttt{#1}}}

\newcommand{\para}[1]{\noindent\textit{#1.}}

\newcommand{\codelink}[2]{
  \href{#2}{\colorbox{gray!15}{{\fontfamily{fvm}\selectfont\small\textcolor{black}{#1}}}}
}

\newtheorem{manualtheoreminner}{Theorem}
\newenvironment{manualtheorem}[1]{
  \IfBlankTF{#1}
    {\renewcommand{\themanualtheoreminner}{\unskip}}
    {\renewcommand\themanualtheoreminner{#1}}
  \manualtheoreminner
}{\endmanualtheoreminner}

\def\BibTeX{{\rm B\kern-.05em{\sc i\kern-.025em b}\kern-.08em
    T\kern-.1667em\lower.7ex\hbox{E}\kern-.125emX}}

\begin{document}

\date{}

\title{\sysname: A Stable Foundation for Modular Network Control Plane Verification \\
(extended version)
}

\author{\IEEEauthorblockN{
Dexin Zhang\orcid{0009-0005-9009-9503}\textsuperscript{*},
Timothy Alberdingk Thijm\orcid{0000-0003-1758-5917}\textsuperscript{\dag},
David Walker\orcid{0000-0003-3681-149X}\textsuperscript{*},
Aarti Gupta\orcid{0000-0001-6676-9400}\textsuperscript{*}}
\IEEEauthorblockA{
\textsuperscript{*}\textit{Princeton University}, Princeton, USA, \{dz7903,dpw,aartig\}@cs.princeton.edu\\
\textsuperscript{\dag}tthijm.fmcad26@gmail.com
}
}

\maketitle

\begin{abstract}
Network operators are often interested in verifying  \emph{eventually-stable properties} of network control planes: properties of control plane states that hold eventually, and hold forever thereafter, provided the operating environment remains unchanged. 
In this work, we introduce \sysname, a new modular framework for verifying such properties, based on the key idea of a \emph{converges-before graph} (\cbgraph for short).
Given user-provided specifications for each network component, \sysname checks the necessary requirements in parallel using an SMT solver and automatically \cbgenerates a \cbgraph. Verification succeeds if all requirements are met and the \cbgenerated \cbgraph connects all nodes in a network. We formalize our verification algorithm in the Lean theorem prover and prove its soundness. Moreover, the \cbgenerated \cbgraph can be used to determine fault tolerance properties of the network. We also consider the inverted problem: Given a connected \cbgraph, we can automatically synthesize suitable component specifications by formulating the problem as a Constrained Horn Clause (CHC) problem that ensures the given correctness property.
We evaluate the performance of \sysname and the synthesis tool (along with comparison to other tools) to demonstrate that they can successfully verify expressive properties on a range of benchmarks.

\end{abstract}


\section{Introduction}
\label{sec:intro}

Network control planes define the paths to forward traffic from routers to destinations, using distributed routing protocols (such as BGP, OSPF, RIP, and others) that are configured by network operators.
In today's large cloud provider networks, these configuration files can be extremely challenging to update correctly due to the network's scale and complexity~\cite{rogers-outage-2022}.
To verify that updates are correct and prevent network outages, researchers have developed control plane verification techniques~\cite{rcdc,bagpipe,minesweeper,arc,abhashkumar2020tiramisu,prabhu2020plankton,ye2020accuracy,beckett2018control,beckett2019abstract,acorn,expresso} (discussed in detail, \S\ref{sec:related}).
However, they often fail to scale on large networks, due to the inherent verification complexity.

Consequently, there is growing interest in \emph{modularly} verifying the correctness of networks' control planes~\cite{rcdc,kirigami,timepiece,lightyear}.
In modular network verification, a correctness property like
\emph{eventually-stable reachability}
(``router \textit{X} eventually always has a path to a destination \textit{D}'') is decomposed into subproblems on the components of the network.
These subproblems check that each component is correct with respect to a suitable \emph{\interface}, assuming that all other components also respect their \interfaces.
This modular approach scales much better than monolithic verification, but requires users to supply correct \interfaces for each component.
Writing such a \interface can be challenging and burdensome: for example, Lightyear~\cite{lightyear} requires users to specify invariants \emph{and} a network path (along which routing messages flow) to prove a liveness property;
Timepiece~\cite{timepiece} requires users to specify a function describing how each component's routing state changes at each time instant, and applies only to synchronous schedules. 

We have developed a new system \sysname for modular network verification with distinct advantages over prior work. 
Unlike Lightyear~\cite{lightyear}, properties that depend on route preferences can be proven by \sysname, and users do not have to supply a path -- a supporting graph is inferred automatically. 
Unlike Timepiece~\cite{timepiece}, \sysname is sound for all schedules, and users need not refer to exact times in their specifications.
In addition, \sysname automatically infers the degree of fault tolerance of any verified property, a key feature of interest to network engineers, not provided by other modular systems.

The approach in \sysname is to focus on \emph{eventually-stable} correctness properties of network behavior, \IE properties that hold at some future time and forever after that%
\footnote{In linear temporal logic, these properties can be expressed as $FGp$, where $F$/$G$ are eventually/globally temporal operators and $p$ is propositional.}.
In addition to properties, every network router (node) $v$ is associated with two \interfaces, $I(v)$ and $Q(v)$. (These can be regarded as per-component annotations, which are common in deductive program verification~\cite{dafny,verus} in the form of per-method annotations.)
Specifically, $I$ overapproximates the set of states the node can enter at any time during the routing decision process, while $Q$ overapproximates its eventually-stable states. Intuitively, $I$ is similar to an invariant, and we call $Q$ an \emph{abstract convergence set} because it overapproximates the states to which the network node converges.

The key idea in \sysname is to automatically infer a \emph{converges-before graph (\cbgraph)}\footnote{This name is inspired by the standard "happened-before relationship"~\cite{lamport-clock}.}, which serves to provide a witness for the network convergence process. 
Rather than require a user to provide a \cbgraph, \sysname automatically generates it from the user-supplied annotations ($I, Q$).
For verification to succeed, the generated \cbgraph \emph{must connect} all nodes to the origin(s) of the routing process.
We have formalized this theory in Lean~\cite{lean} and proved that our algorithm is sound with respect to a standard asynchronous network semantics.

In addition, the \cbgraph \cbgenerated by \sysname has the maximal connectivity for the given $I$ and $Q$: if edges of a \cbgraph fail, then connectivity is reduced.
Hence, by analyzing the connectivity of the \cbgraph, it is possible to determine the degree of fault tolerance of any proven property---a key concern for network engineers building resilient systems. 

Finally, we also consider the inverted problem: when a connected \cbgraph is user-supplied, can suitable component specifications $I$ and $Q$ be synthesized automatically for a given property? We use our theory to formulate this as a CHC (Constrained Horn Clause) problem~\cite{ChcFlc15}, and use an off-the-shelf CHC solver for demonstrating this capability.

\para{Summary of Contributions}
\begin{itemize}
\itemsep0em
\item \emph{Modular control plane verification}: We present the theory (\S\ref{sec:theory}) and algorithm for \sysname (\S\ref{sec:tool}). Given \interfaces ($I, Q$) and target properties, it generates and checks per-node and per-edge verification conditions (VCs) in parallel using Satisfiability Modulo Theory (SMT) solvers~\cite{z3}, and performs a lightweight global graph connectivity check to prove correctness.
\item \emph{\cbgraphs for sound modular reasoning} (\S\ref{subsec:cbgraphs}): \sysname automatically \cbgenerates a \cbgraph, which provides a structure for well-founded induction in reasoning about correctness. We formalize our theory of abstract convergence and the verification algorithm in Lean~\cite{lean} and prove its soundness.
\item \emph{Automated fault analysis via \cbgraphs} (\S\ref{subsec:fault-tolerance}): We show that if the \cbgenerated \cbgraph is connected despite $k$ arbitrarily selected \cbgraph edge failures, then the verified property is $k$-fault tolerant.
\item \emph{Automated synthesis of \interfaces}   (\S\ref{sec:synth}): When a connected \cbgraph is given, we can automatically synthesize suitable $I$ and $Q$ for each node by formulating a CHC problem based on our theory.  
\item \emph{Evaluation} (\S\ref{sec:eval}): We have developed prototype implementations of \sysname and the synthesis tool (\synthname), and evaluate them on a range of benchmarks, illustrating the expressiveness of the properties, the performance, and comparison with other tools. 
\end{itemize}

In comparison to the most closely related modular verifiers (more details in \S\ref{sec:related}), our framework focuses on a broad, new category of eventually-stable properties, has different (simpler, we would argue) \interfaces, adds failure analysis, and comes with an additional tool for \interface synthesis.
\section{\sysname: The Theory}
\label{sec:theory}

\subsection{Network Model and Semantics}

The \emph{network control plane} is a set of \emph{routers} that exchange messages with each other about available paths (routes) to a destination using standard \emph{routing protocols} such as BGP, OSPF, RIP~\cite{bgp,ospf,rip}.
Each router maintains a route that is most preferred among all available ones from its neighbors (with the preference relation depending on the protocol).  
Routers can modify or drop routes from their neighbors during message transfer, based on the configuration of export and import policies.
\footnote{In our model, we use the term "route" to refer to both states and messages.}
They repeatedly and asynchronously exchange routes, and the network control plane (hopefully) converges with every router selecting its best route to a destination. 

We represent network control planes using a standard network model, inspired by routing algebras~\cite{griffin2005metarouting,sobrinho2005algebraic},
where each network instance considers one destination prefix (\IE{} a set of IP addresses).
This model has been applied to a wide range of routing protocols used in practice, including BGP, RIP, OSPF.

Formally, a \emph{network instance} is a 5-tuple $N=(G,R,\mathrm{init},f,\oplus)$, where:
\begin{itemize}
    \itemsep0em
    \item $G=(V,E)$ is a simple, finite directed graph. The nodes $V$ are routers and directed edges $E \subseteq V \times V$ are links.
    \item $R$ is the set of routing messages, or simply \emph{routes}, that can be exchanged between nodes. $R$ varies according to the protocol; $\infty\in R$ (invalid route) is used when no route is selected.
    \item $\mathrm{init}:V\to R$ is a function that maps each node $v \in V$ to its initial route $\mathrm{init}(v) \in R$, which is the route that originates from $v$ or is imported from external sources at a border node. It is $\infty$ if there are no routes originated or imported at node $v$.
    \item $f:E\to(R\to R)$ maps each edge $(u,v)\in E$ to a transfer function $f(e):R\to R$ (also written as $f_e$). For instance, $f_{(u,v)}$ is the transfer function for routes sent from $u$ to $v$, based on the configuration of export and import policies. 
    \item $\oplus:R\times R\to R$ is a binary associative, commutative and selective function (the \emph{merge} function) that selects the best route between two routes. It is defined by each protocol; $\infty$ is required to be the least preferred route.
\end{itemize}

\para{Encoding network model in SMT}
Modeling routing protocols using SMT theories is a well-understood process~\cite{minesweeper}. The set of routes $R$ is represented as a record of fields, which depend on the specific protocol, and each field is encoded as a sort (i.e., type) in some first-order theory~\cite{BarFT-SMTLIB}.
The other components --- $\mathrm{init}, f, \oplus$ --- are modeled as expressions and formulas in first-order theories.
Full details of our SMT encoding of the BGP protocol are available in Appendix~\ref{appendix-subsec:smt}.

\para{Asynchronous schedules and fairness}
To model an asynchronous network, we follow the approach of Üresin and Dubois~\cite{uresin1990parallel} and Daggitt et al.~\cite{daggitt2018asynchronous}. 
In this model, times are natural numbers $\mathbb{T}=\mathbb{N}$, and a \emph{schedule} for a network $N$ is a pair $(\alpha,\beta)$ with:
\begin{enumerate}
    \itemsep0em
    \item An activation function $\alpha:\mathbb{T}\to2^{V}$ that maps a time $t$ to a set of routers $\alpha(t)$ activated at that time.
    \item A data flow function $\beta:E\to(\mathbb{T}\to\mathbb{T})$ that defines
    the time taken to traverse each link $e=(u,v)$. 
    If $t_1=\beta(e)(t_2)$ (also written $\beta_e(t_2)$), the route sent by $u$ at time $t_1$ will arrive at $v$ at time $t_2$. Messages must be received after they are sent and hence $t_2>t_1$ in all cases. 
\end{enumerate}

These asynchronous schedules can model arbitrary delay, loss, duplication,  or reordering of any message in transit.
When proving soundness, we additionally assume that nodes and edges are not \emph{failed}: they activate infinitely often
and send/receive up-to-date messages infinitely often (formal definitions are presented in Appendix~\ref{appendix-subsec:proof}).
Such schedules with no failures are called \emph{fair schedules}.
In \S\ref{subsec:fault-tolerance}, we consider tolerance under edge failures; elsewhere, we assume a schedule is fair unless stated otherwise.

\para{Network semantics and concrete convergence} The \emph{network semantics} of $N$ with respect to a schedule $S=(\alpha,\beta)$ is defined as the following function $\sigma_S:V\times\mathbb T\to R$. (This is the same as in prior work on routing algebras~\cite{daggitt2018asynchronous}.)
\begin{equation*}
\sigma_S(v,t)=
\begin{cases}
\mathrm{init}(v) & t=0 \\
\sigma_S(v,t-1) & t>0\land v\notin\alpha(t) \\
\multispan2{$\mathrm{init}(v)\oplus\bigoplus_{(u,v)\in E}f_{u,v}(\sigma(u,\beta_{u,v}(t)))$\hfill} \\
\phantom{\bigoplus_{(u,v)\in E}f_{u,v}(\sigma(u,\beta_{u,v})} & t>0\land v\in\alpha(t)
\end{cases}
\end{equation*}
At $t=0$, each node $v$ selects its initial route $\mathrm{init}(v)$.
Henceforth, if $v$ is activated at time $t$ (\IE $v\in\alpha(t)$), it will select the best route among all routes received from its neighbors (possibly delayed). Otherwise, its route remains the same as at time $t-1$.

For a given schedule $S$ and network $N$, 
we say node $v$ \emph{(concretely) converges to} $r$ if
$\exists \tau.\ \forall t\ge\tau,\ \sigma_S(v,t)=r$.  In other words,
$v$ converges to $r$ if $\sigma_S(v,t)$ is equal to $r$ for some $t$ and forever after.
At the network level, we say the network semantics $\sigma_S$ \emph{converges to} $\sigma_S^*:V\to R$ if $v$ converges to $\sigma_S^*(v)$ for each node $v$. Such a state $\sigma_S^*$ is called a converged state of the network.
Note that a network may have multiple converged states for different schedules $S$ (this 
can happen in real networks~\cite{batfish-lessons}), or it may have no converged state (\EG it may oscillate between a set of states~\cite{griffin2002stable}).

\subsection{Running Example and Correctness Properties}

\begin{figure}
    \centering
    \scalebox{0.8}{
    \begin{tikzpicture}
        \node[node1](A){$A$};
        \node[node2,above right=0.5cm and 1cm of A](B){$B$};
        \node[node2,below right=0.5cm and 1cm of A](C){$C$};
        \node[node2,below right=0.5cm and 1cm of B](E){$E$};
        \draw[-](A)--(B);
        \draw[-](A)--(C);
        \draw[-](B)--(E);
        \draw[-](C)--(E);
    \end{tikzpicture}
    }
    \caption{\label{fig:eg1}Running Example: a toy network with four nodes; $A$ is the destination and originates a route to itself; local preference is set to 300 when $E$ imports a route from $B$.}
\end{figure}

Consider a simple network, shown in Fig.~\ref{fig:eg1}, with four nodes running the BGP routing protocol. A network operator would like to prove that node $E$ has a route to destination node $A$ after the network has converged, 
and moreover that the route goes through $B$ rather than $C$ (perhaps the former route is less costly than the latter). 
To implement this preference, the network operator uses BGP's
local preference mechanism, setting the local preference attribute to 300 when $E$ imports routes from $B$,
and leaving it as the default value 100 everywhere else.
When node $E$ receives routes from multiple neighbors, it chooses the route with the highest local preference (in this case, the route from $B$ with value 300 over the route from $C$ with value 100).  
The routing process starts with node $A$ \emph{originating} a route by announcing a route to itself, and no other node having an initial route to $A$.
In the rest of the paper, we will use this as a running example to illustrate key ideas and concepts in our theory.

\para{Formal definition of correctness properties} 
A user-provided correctness property $Y:V\to 2^R$ maps each node $v$ to some predicate (formula) $Y(v)$ on routes $R$. 
Our theory supports verification of \emph{eventually-stable} properties associated with such state predicates $Y$. More formally, we verify that for all schedules $S$:
$$\forall v.\ \exists\tau.\ \forall t\ge\tau.\ \sigma_S(v,t)\in Y(v).$$

\para{Running example}
In our example, suppose we have the following properties of interest:
\begin{itemize}
    \item $Y_1(E)$: Will node $E$ be able to reach $A$?
    \item $Y_2(E)$: Will node $E$ select a path to $A$ using a route from node $B$ rather than from node $C$? 
\end{itemize}
The first property $Y_1(E)$ expresses a simple reachability property, while the second $Y_2(E)$ captures a preference property.
The user is interested in checking that these properties hold at some point and forever after (provided changes to the operating environment do not occur, e.g., configuration/policy changes).
To express these properties in terms of a route state $s$ at $E$,
we define a routing field $\aspath(s)$ that represents the AS path of $s$, represented as a list of AS numbers of routers in the routing path so far: 
\begin{align*}
Y_1(E)&=\{s\mid s\ne\infty\} \\
Y_2(E)&=\{s\mid s\ne\infty\land C\notin\aspath(s)\}
\end{align*}

\noindent Note that the predicates here refer implicitly to the destination $A$, \IE{} all routes $s$ are interpreted as routes to $A$ (and $s \ne\infty$ represents some available route to $A$). 

\para{Expressiveness of Properties} \sysname requires properties $Y$ to be written as per-node predicates over routes $R$, to enable modular verification.
However, global properties or path-dependent properties can also be expressed in \sysname in a modular form through decomposition or the introduction of ghost variables and ghost code. For example, global connectivity (\IE connectivity between all or certain nodes) can be decomposed into single-destination connectivity per destination $d$. For each $d$, we can define property $Y(v)$ to be ``$v$ can reach the destination $d$'' (as shown in the running example above for a single destination $A$), and then verify separately for all/certain destinations $d$.
For path-dependent properties, we can introduce ghost fields that are updated along the path. 
For example, to verify that any route selected by $v$ must go through a node in the set of nodes $W$, one can add a Boolean ghost field $\mathrm{ViaW}$ that is set to true when a route is imported by a node in $W$. 
We can then define $Y(v) = \{ s \mid \mathrm{ViaW}(s) = \mathrm{true}\}$.
Use of ghost fields is also common in deductive program verification~\cite{dafny,verus}.
We note that properties beyond path-dependent properties, such as relational properties between arbitrary routers, may require some heavy instrumentation (\EG via a product construction).
For example, a guarantee that router A selects ACE if and only if router B selects BDE is a relational property.

\subsection{Abstracting Convergence for Proving Correctness}
\label{subsec:interfaces}

\sysname requires the user to provide an annotation $Q$, which we call an \emph{abstract convergence set}. 
Formally, $Q: V\to 2^R$, which maps each node $v$ to some predicate (formula) $Q(v)$ on routes. 
Given a specific network instance and (fair) schedule $S$, we say that $\sigma_S(v,t)$ \emph{abstractly converges to $Q(v)$} if:
$$\exists \tau.\ \forall t\ge \tau,\ \sigma_S(v,t)\in Q(v).$$

Note that if $\sigma_S(v,t)$ abstractly converges to $Q(v)$, this schedule may not actually converge to a \emph{single} concrete state in $Q(v)$, 
but if it does, then the concrete converged route $\sigma_S^*(v)$ must be included in the set $Q(v)$.
We say \emph{a node $v$ abstractly converges to $Q(v)$} if $\sigma_S(v,t)$ abstractly converges to $Q(v)$ for all schedules $S$.

In the rest of the paper, we say ``abstractly converges to $Q(v)$'' to signify that we are interested in reasoning about convergence to a \emph{set of states}, rather than convergence to a single state. 
This abstraction of convergence via an overapproximation is novel to our approach and provides crucial benefits.
The user does not have to specify \emph{exactly} which state a node converges to, only a set containing the concrete state(s); moreover,
if a network does not actually converge to a single state, but loops within a subset of the states characterized by $Q$, 
violations of expected properties in such states will still provide useful feedback to users. 
(Other related work on guaranteeing (unique) convergence in routing protocols~\cite{daggitt2018asynchronous,daggitt2018rate} is orthogonal to our work, see more details in \S\ref{sec:related}. \OMIT{more details in \S\ref{sec:related}.})

Our two-step strategy for proving correctness is to (1) establish that all nodes $v$ abstractly converge to $Q(v)$, 
and (2) for each node $v$, check that every route in $Q(v)$ satisfies $Y(v)$. If any step fails, then verification fails. 
The second step is easy to do: we use an SMT solver to check the validity of $Q(v) \Rightarrow Y(v)$ (where $\Rightarrow$ denotes propositional implication). 
However, the first step is more challenging, and is discussed in the next subsection.

\para{Running example}
We consider the second step for our running example. To prove $Y_1(E)$, it suffices to choose $Q_1(v) = \{s\mid s\ne\infty\}$ for $v\in\{A,B,C,E\}$.
Note that $Q_1(E)=Y_1(E)$ here, so the second step holds trivially. 
However, to prove $Y_2(E)$, we need a stronger \interface $Q_2(E)$ ($\lp(s)$ represents local preference of the route $s$, and $\len(s)$ represents the length of the AS path of $s$):

\vspace*{-0.15in}
{\small
\begin{equation*}
    Q_2(E)=\{s\mid s\ne\infty\land\lp(s)=300\land\len(s)=2\land C\notin\aspath(s)\}
\end{equation*}}

Furthermore, the \interfaces $Q_2(v)$ at other nodes $v$ must also be more specific to affirm the absence of node $C$ along the route transferred from $A$ through $B$ to $E$:

\vspace*{-0.15in}
{\small
\begin{align*}
    Q_2(A)&=\{s\mid s\ne\infty\land\lp(s)=100\land\len(s)=0\land C\notin\aspath(s)\} \\
    Q_2(B)&=\{s\mid s\ne\infty\land\lp(s)=100\land\len(s)=1\land C\notin\aspath(s)\} \\
    Q_2(C)&=\{s\mid s\ne\infty\land\lp(s)=100\land\len(s)=1\}
\end{align*}
}

Given the above $Q(v)$ \interfaces for each node $v$, it is straightforward to check that for each property, $Q_1(E) \Rightarrow Y_1(E)$ and $Q_2(E) \Rightarrow Y_2(E)$.

\subsection{\cbgraphs: Key Structure for Sound Reasoning}
\label{subsec:cbgraphs}

Recall that we still need to establish the first step: that all nodes $v$ abstractly converge to their given specifications $Q(v)$, \IE $Q(v)$ is a \emph{sound} abstract convergence set for $v$.
To help with this, a user also provides a specification $I(v)$ for each node $v$ --- we call $I$ an \emph{invariant}. 
To show that $I(v)$ is sound, we check that it overapproximates the set of routes that node $v$ selects at any time during the routing process. 

\sysname performs additional checks to ensure that the user-provided specifications $I(v)$ and $Q(v)$ are sound. A justification at node $v$ could depend on the specifications of its neighbor nodes, but not if they, in turn, are justified recursively by the specifications at $v$. This circular reasoning would clearly be unsound. 

To ensure soundness of our theory, we introduce the \cbgraph.
Intuitively, a \cbgraph provides a witness that all nodes $v$ in the network abstractly converge to their $Q(v)$.
\OMIT{a witness for all nodes $v$ in the network to abstractly converge to their $Q(v)$.}
Formally, a \cbgraph consists of \cbroots and \cbedges, which are defined as follows:
\begin{itemize}
    \itemsep0em
    \item \emph{\cbroots}: a non-empty subset of nodes $w \in V$, where each $w$ must select a route $\sigma_S(w,t)\in Q(w)$ at any time $t$ (starting from $t=0$), for any schedule $S$. 
    In practice, a node that originates routes as a destination, or acts as a border node (that imports routes from external sources) is a \cbroot.
    
    \item \emph{\cbedges}: a subset of directed edges $(u,v)$ in the network, such that for any schedule $S$, if $u$ selects a route $\sigma_S(u,t_1)\in Q(u)$ at time $t_1$, and transfers this route along edge $(u,v)$, then $v$ \emph{must} select a route $\sigma_S(v,t_2)\in Q(v)$, where $t_2$ is the time that $v$ receives and merges $u$'s route.
\end{itemize}

\para{Running example}
Fig.~\ref{fig:example-cbgraphs} shows the \cbgraphs \cbgenerated by \sysname for our running example with \interfaces $I_1,Q_1$ and $I_2,Q_2$, respectively. The destination node $A$ is a \cbroot in both, and the red directed arrows show the \cbedges. 
The node $A$ is a \cbroot because it originates a route: \OMIT{due to its originated routes:} the routes at $A$ at $t=0$ and at all times thereafter belong to the \interfaces $Q_1(A)$ and $Q_2(A)$.

\begin{figure}[t!p]
    \centering
    \begin{subfigure}[b]{0.23\textwidth}
        \scalebox{0.8}{
        \begin{tikzpicture}
            \node[node1](A){$A$};
            \node[node2,above right=0.5cm and 1cm of A](B){$B$};
            \node[node2,below right=0.5cm and 1cm of A](C){$C$};
            \node[node2,below right=0.5cm and 1cm of B](E){$E$};
            \draw[-](A)--(B);
            \draw[-](A)--(C);
            \draw[-](B)--(E);
            \draw[-](C)--(E);
            \node[above left=0.1cm and -0.6cm of A]{\color{red}{\cbroot}};
            \draw[color=red,shorten >= 0.25cm,shorten <= 0.1cm,transform canvas={yshift=0.2cm},->](A)--(B);
            \draw[color=red,shorten >= 0.25cm,shorten <= 0.1cm,transform canvas={yshift=-0.2cm},->](B)--(A);
            \draw[color=red,shorten >= 0.1cm,shorten <= 0.25cm,transform canvas={yshift=0.2cm},->](B)--(E);
            \draw[color=red,shorten >= 0.1cm,shorten <= 0.25cm,transform canvas={yshift=-0.2cm},->](E)--(B);
            \draw[color=red,shorten >= 0.25cm,shorten <= 0.1cm,transform canvas={yshift=-0.2cm},->](A)--(C);
            \draw[color=red,shorten >= 0.25cm,shorten <= 0.1cm,transform canvas={yshift=0.2cm},->](C)--(A);
            \draw[color=red,shorten >= 0.1cm,shorten <= 0.25cm,transform canvas={yshift=-0.2cm},->](C)--(E);
            \draw[color=red,shorten >= 0.1cm,shorten <= 0.25cm,transform canvas={yshift=0.2cm},->](E)--(C);
        \end{tikzpicture}
        }
        \caption{\cbgraph for $Q_1$}
    \end{subfigure}
    ~
    \begin{subfigure}[b]{0.23\textwidth}
        \scalebox{0.8}{
        \begin{tikzpicture}
            \node[node1](A){$A$};
            \node[node2,above right=0.5cm and 1cm of A](B){$B$};
            \node[node2,below right=0.5cm and 1cm of A](C){$C$};
            \node[node2,below right=0.5cm and 1cm of B](E){$E$};
            \draw[-](A)--(B);
            \draw[-](A)--(C);
            \draw[-](B)--(E);
            \draw[-](C)--(E);
            \node[above left=0.1cm and -0.6cm of A]{\color{red}{\cbroot}};
            \draw[color=red,shorten >= 0.1cm,shorten <= 0.25cm,transform canvas={yshift=-0.2cm},->](A)--(B);
            \draw[color=red,shorten >= 0.25cm,shorten <= 0.1cm,transform canvas={yshift=-0.2cm},->](B)--(E);
            \draw[color=red,shorten >= 0.1cm,shorten <= 0.25cm,transform canvas={yshift=0.2cm},->](A)--(C);
        \end{tikzpicture}
        }
        \caption{\cbgraph for $Q_2$}
    \end{subfigure}
    \caption{\cbgraphs \cbgenerated for the running example.}
    \label{fig:example-cbgraphs}
\end{figure}

As an example of a \cbedge, consider the edge $BE$ in the network with the $I_1,Q_1$ \interfaces. If node $B$ has a route in $Q_1(B)$ (\IE{} if it has a route at all), then a route transfer from $B$ to $E$ will definitely result in $E$ having a route in $Q_1(E)$ \emph{after} the transfer. Critically, this is true regardless of which other messages may have reached $E$ earlier (i.e., for any state in $I_1(E)$ that node $E$ was in earlier). Thus, the edge $BE$ is a \cbedge. In fact, every directed edge in the network is a \cbedge for $I_1,Q_1$. This is because no routes are dropped along any edge; thus, transfer of a route along each edge necessarily establishes $Q_1$ at the receiving node.

As an example of an edge that is \emph{not} a \cbedge, consider the edge $CE$ for the $I_2,Q_2$ \interfaces.
Suppose node $C$ has a route in $Q_2(C)$, and transfers this route to node $E$. Suppose $E$ does not have any route before this transfer, 
perhaps due to delays along the $BE$ edge. 
Then, as a result of transfer from $C$, node $E$ will select this route where the path includes $C$. However, this route  \emph{does not belong to $Q_2(E)$}, as defined earlier (since its path includes $C$). Therefore, the $CE$ edge is not a \cbedge. 

\para{Connected \cbgraphs}
We say that a \cbgraph is \emph{connected} if every node $v$ in the network is itself a \cbroot or has a path of \cbedges from some \cbroot $w$ to $v$.
The following important theorem states that a \emph{connected \cbgraph} provides a basis for reasoning about convergence. 
(A proof sketch of the theorem and Lean formalizations appear in Appendix~\ref{appendix-subsec:proof}; a formal proof has been mechanized in Lean.\footnote{Our Lean proof is available at \url{https://github.com/dz7903/cbgraphs-formalization}.})
\OMIT{is done in Lean.}

\begin{theorem}[Connected \cbgraph Theorem]
    \label{thm:connected-cbgraph}
    If a \cbgraph for $Q$ is connected, then all nodes $v$ in the network abstractly converge to $Q(v)$. 
\end{theorem}

Note that we do not require a \cbgraph to be acyclic -- this provides flexibility while ensuring soundness.
Any connected CB-graph must have a connected acyclic subgraph that can serve as a witness for all nodes $v$ in the network to abstractly converge to their $Q(v)$.

More importantly, we note that \sysname does not require a user to provide a \cbgraph to establish abstract convergence; it automatically \cbgenerates a \emph{maximal} \cbgraph and checks that it is connected.  
More details are presented in \S\ref{subsec:alg1}, where \sysname maximally identifies each node as a \cbroot and each edge as a \cbedge, provided they meet certain requirements.

\subsection{\cbgraphs and Fault tolerance}
\label{subsec:fault-tolerance}

In this subsection, we consider edge failures and the correctness of properties under at most $k$ edge failures. An edge is considered failed if it does not deliver any messages sent after a certain time (a formal definition is shown in the Lean formalizations in Appendix~\ref{appendix-subsec:proof}ß).
We say that a \cbgraph is \emph{$k$-fail-connected}, if it remains connected regardless of which combination of $k$ \cbedges is removed.
\begin{theorem}[Fault tolerance theorem]
    \label{thm:fail-connected-cbgraph}
    If a \cbgraph for $Q$ is $k$-fail-connected, then for any schedule $S$ where at most $k$ edges are failed, all nodes $v$ in the network abstractly converge to $Q(v)$ and the network $N$ is \emph{$k$-fault tolerant}.
\end{theorem}

\begin{figure}[t]
    \centering
    \includegraphics[width=\linewidth]{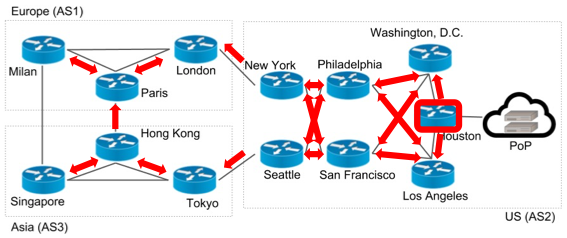}
    \caption{\cbgraph of a cross-world network from the Batfish tutorial~\cite{batfish-tutorials-failure}. Houston (bordered in red) is a \cbroot and directed arrows are \cbedges.}
    \label{fig:fault-tolerance-network}
\end{figure}

\para{Example for fault-tolerance} Consider Fig.~\ref{fig:fault-tolerance-network}, which shows a network with 13 routers and 3 ASs (taken from a Batfish tutorial~\cite{batfish-tutorials-failure}).
Suppose a user is interested in checking reachability to Houston, and provides $Q$ such that it includes any route except $\infty$. 
The \cbgenerated \cbgraph is also shown in the figure.
This \cbgraph does not tolerate any 1-edge failure, because the removal of the \cbedge from Seattle to Tokyo makes Hong Kong disconnected from the \cbroot Houston (note that the \cbedge from Hong Kong to Paris is one-directional).

\section{From Theory to Modular Tool: \sysname}
\label{sec:tool}
\label{subsec:verification-algorithm}

We now describe how we use our theory of abstract convergence to develop a modular verifier \sysname. 
\begin{figure*}[t]
    \centering
    \small
    \begin{tabular}{llll}
        \toprule
        VC name & VC Set & VC Formula (per node $v$, per edge $(u,v)$) & Explanation \\
        \midrule
        Initial &
        $\vcinit(v)$ &
        $\mathrm{init}(v)\in I(v)$ &
        \parbox[c]{6.5cm}{The \Iinterface $I(v)$ holds for the initial route of $v$.} \\
        \midrule
        Invariance &
        $\vcinv((u,v))$ &
        $\begin{aligned}
            &\forall s_u,s_v.\ s_u\in I(u)\land s_v\in I(v) \\
            &\quad\Rightarrow s_v\oplus f_{(u,v)}(s_u)\in I(v)
        \end{aligned}$ &
        \parbox[c]{6.5cm}{The \Iinterface $I(v)$ continues to hold when $v$ receives and merges a route from its neighbor $u$.} \\
        \midrule
        Property &
        $\vcprop(v)$ &
        $\forall s_v.\ s_v\in Q(v)\Rightarrow s_v\in Y(v)$ &
        \parbox[c]{6.5cm}{The \Qinterface $Q(v)$ implies the property $Y(v)$.} \\        
        \midrule
        \cbroot &
        $\vcroot(v)$ &
        $\begin{aligned}
            &\mathrm{init}(v)\in Q(v) \\
            &\land\bigwedge_{\substack{u\in V \\ (u,v)\in E}}\left(\begin{aligned}
                &\forall s_u,s_v.\ s_u\in I(u)\land s_v\in Q(v) \\
                &\quad\Rightarrow s_v\oplus f_{(u,v)}(s_u)\in Q(v)
            \end{aligned}\right)
        \end{aligned}$ &
        \parbox[c]{6.5cm}{For a \cbroot $v$, the \Qinterface $Q(v)$ holds for its initial route, and continues to hold when $v$ receives and merges a route from its neighbor $u$.} \\
        \midrule
        \cbedge &
        $\vcedge((u,v))$ &
        $\begin{aligned}
            &\forall s_u,s_v.\ s_u\in Q(u)\land s_v\in I(v) \\
            &\quad\Rightarrow s_v\oplus f_{(u,v)}(s_u)\in Q(v)
        \end{aligned}$ &
        \parbox[c]{6.5cm}{For a \cbedge $(u,v)$, if $u$ has converged to its \Qinterface $Q(u)$ and sends its converged route to $v$, $v$ must satisfy its $Q(v)$ after receiving and merging it.} \\
        \bottomrule
    \end{tabular}
    \caption{Summary: Verification Condition (VC) Formulas with Explanations.}
    \label{fig:VCs}
\end{figure*}

\subsection{\sysname Algorithm}
\label{subsec:alg1}

\sysname generates and checks Verification Conditions (VCs) on nodes and edges in the network. These are summarized in Fig.~\ref{fig:VCs}, where each row shows a different VC set based on their role (explained below), and the respective formula generated by \sysname.
The formulas capture conditions on arbitrary routes $s_u$ and $s_v$
at nodes $u$ and $v$, respectively.  The conditions are expressed in terms of the network model (\IE with operations $init, f, \oplus$ for the network instance), \interfaces ($Q, I$), and correctness property ($Y$) provided by the user. 
For each of these VCs, \sysname invokes an SMT solver (\EG Z3~\cite{z3}) to check its validity. If the check fails, the solver provides a counterexample, \IE values of $s_u, s_v$ where the formula is false. 

The first three VCs are regarded as \emph{essential} VCs, because if any of them fail, then \sysname immediately reports failure.
They serve the following roles:
\begin{itemize}
    \itemsep0em
    \item
    \textbf{Checking \Iinterface} $I$:
    $\vcinit$ and $\vcinv$ together ensure that the user-provided \interface $I$ is an inductive invariant at each node $v$. As shown, $\vcinit$ performs a base check on every node $v$, to ensure that the initial route $\mathrm{init}(v)$ is in $I(v)$.
    $\vcinv$ performs an inductive check for every edge $(u,v)$: 
    if a route $s_u$ is in $I(u)$,
    then the route $s_v$ selected by $v$ after its transfer-and-merge must also be in $I(v)$. 
    Here, the subformula $s_v\oplus f_{u,v}(s_u)$ represents the new route at node $v$ after the merge ($\oplus$) of its existing route $s_v$ and the transferred route ($f_{(u,v)}(s_u))$ from neighbor $u$.  
    \item \textbf{Checking correctness property} $Y$: $\vcprop$ ensures that every route in the
    \Qinterface
    $Q(v)$ satisfies $Y(v)$. This is needed for checking the correctness of the eventual-stability property.
\end{itemize}

\sysname requires that a user provide $Q, I$, but it alleviates some burden by constructing a \cbgraph for $Q$. 
It checks two VCs, $\vcroot$ and $\vcedge$, on all nodes and edges to identify \cbroots and \cbedges, respectively. 
\begin{itemize}
    \itemsep0em
    \item \textbf{Identifying \cbroots}: $\vcroot(v)$ checks that $\mathrm{init}(v)$ is in $Q(v)$, and for every incoming edge $(u,v)$ with $s_u$ as some arbitrary selected state at node $u$ (\IE $s_u \in I(u)$), the route selected by $v$ after transfer-and-merge must also belong to $Q(v)$. 
    If $\vcroot(v)$ is valid, then node $v$ is a \cbroot.  
    \item \textbf{Identifying \cbedges}: $\vcedge((u,v))$ checks that if $s_u$ is in $Q(u)$, and if $s_v$ is some arbitrary selected route at node $v$ (\IE $s_v \in I(v)$), then the route selected by $v$ after transfer-and-merge must belong to $Q(v)$. 
    If $\vcedge((u,v))$ is valid, then edge $(u, v)$ is a \cbedge.
\end{itemize}

\begin{figure}
\centering
\begin{algorithmic}[1]
  \REQUIRE Network $N$, Correctness Properties $Y$, and modular \interfaces $I, Q$
  \ENSURE Correct or Fail (with counterexamples)
  \STATE $\mathrm{CBRoots}\gets\emptyset,\mathrm{CBEdges}\gets\emptyset$
  \FOR{\label{alg:phase1-start} $v\gets V$ \textbf{parallel}}
    \IF {$\vcinit(v)\land\vcprop(v)$ is not valid}
      \RETURN Fail
    \ENDIF
    \IF {$\vcroot(v)$ is valid}
      \STATE $\mathrm{CBRoots}\gets\mathrm{CBRoots}\cup\{v\}$
    \ENDIF
  \ENDFOR
  \FOR{$e\gets E$ \textbf{parallel}}
    \IF {$\vcinv(e)$ is not valid}
      \RETURN Fail
    \ENDIF
    \IF {$\vcedge(e)$ is valid}
      \STATE $\mathrm{CBEdges}\gets\mathrm{CBEdges}\cup\{e\}$
    \ENDIF
  \ENDFOR
  \label{alg:phase1-end}
  \IF {\label{alg:phase2-start} $\mathrm{IsConnected}(\mathrm{CBRoots},\mathrm{CBEdges})$}
    \RETURN Correct
  \ELSE
    \RETURN Fail
  \ENDIF
  \label{alg:phase2-end}
\end{algorithmic}
\caption{\sysname Verification Algorithm.}
\label{alg:verification}
\end{figure}


\para{Putting it all together}
The complete verification algorithm including generation and checking of the VCs is shown (as pseudocode) in Fig.~\ref{alg:verification}.
It takes a network $N$, correctness property $Y$, and \interfaces $I, Q$ as inputs, and produces output Correct or Fail (with counterexamples).
\OMIT{It takes as inputs: a network $N$, correctness property $Y$, and \interfaces $I, Q$; and produces output Correct or Fail (with counterexamples). }

The algorithm has two phases. 
In phase~1 (lines \ref{alg:phase1-start}-\ref{alg:phase1-end}), \sysname generates and checks all VCs on the nodes and edges in parallel. 
If any of the essential VCs ($\vcinit, \vcinv, \vcprop$) is not valid, then \sysname reports Fail immediately, along with a counterexample for debugging purposes. 
The other two VCs ($\vcroot$, $\vcedge$) are not essential in the sense that \sysname does not report Fail immediately. 
If a VC in $\vcroot$ or $\vcedge$ is not valid, the associated node or edge is \emph{not} added to the corresponding set (in phase 1), and \sysname will save the counterexample for later use.

In phase~2 (line \ref{alg:phase2-start}-\ref{alg:phase2-end}), \sysname uses a standard breadth-first search (BFS) graph algorithm to check whether there exists a connected \cbgraph using the \cbroots and \cbedges identified in phase~1.
If \sysname finds a connected \cbgraph, it reports Correct; 
otherwise, it reports Fail (along with saved counterexamples for unconnected nodes).

\para{Correctness of \sysname}
The correctness (soundness) 
follows from Theorem~\ref{thm:connected-cbgraph} and the fact that
\OMIT{is based on Theorem~\ref{thm:connected-cbgraph} and that}
VCs in Fig.~\ref{fig:VCs} ensure soundness of $I, Q$ and identification of a \cbgraph.
We state the correctness theorem below (a proof sketch and the Lean formalized proof are shown in Appendix~\ref{appendix-subsec:proof}):
\begin{theorem}[Correctness Theorem]
    \label{thm:correctness}
    If the \sysname verification algorithm reports Correct for a network $N$, property $Y$, and \interfaces $I, Q$, then for any 
    schedule $S$, the network semantics $\sigma_S(v,t)$ satisfies $\forall v.\ \exists\tau.\ \forall t\ge\tau.\ \sigma_S(v,t)\in Y(v).$
\end{theorem}

\para{Fault tolerance analysis}
\sysname can support fault tolerance analysis simply by checking the $k$-fail-connectedness of the \cbgenerated  \cbgraph.
Algorithmically, \sysname's phase~2 check executes a variant of Dinitz's algorithm~\cite{Dinitz2006} to find the maximal value of $k$ such that the \cbgraph is $k$-fail-connected.
Theorem~\ref{thm:fail-connected-cbgraph} also applies to \sysname, ensuring the soundness of the fault tolerance analysis.

\para{Completeness of \sysname}
The theorem stated above guarantees that \sysname is sound.
We have shown (via a Lean formalization) that \sysname is also complete when all transfer functions $f$ are
\emph{strictly monotonic} (i.e., $x < y \Rightarrow f(x) < f(y)$, where $x < y$ is defined as $x \ne y$ and $x\oplus y=x$) and $<$ is a well-ordering relation on $R$.

However, \sysname is not complete in general for arbitrary routing protocols.
An example that shows incompleteness, inspired by Daggitt and Griffin's work~\cite{daggitt2018rate}, is presented in Appendix~\ref{appendix-subsec:completeness}ß.
Intuitively, \sysname is complete when the network has a “simple pattern of convergence”, where the routing tree (or an acyclic graph) in the converged state can be used as a \cbgraph.
Our evaluations (\S\ref{sec:eval}) show that \sysname can successfully verify many networks in practice, in which all networks have such a pattern of convergence.

\subsection{Annotation Design and Debugging Methodology}
To use \sysname in practice, a user needs to provide suitable $I$ and $Q$ that can be proven sound, and $Q(v)$ must be sufficient to establish the desired property $Y(v)$.
To specify such $I$ and $Q$, a user typically starts from the property $Y$ and tries specifying $Q$ as a subset of $Y$, \emph{and} also a superset of the converged routes.
This can often be expressed as a conjunction of $Y$ with some strengthening clauses. After finding $Q$, the user can try specifying $I$ as a superset of all possible routes selected at any time -- typically expressed as a disjunction of $Q$ with some "pre-convergence" routes. \footnote{\sysname does not enforce $Q$ to be a subset of $I$, but specifying an $I$ that over-approximates $Q$ is more intuitive.}

Starting from an initial set of annotations, which may be incorrect or inadequate, the user can use an iterative debugging process using \sysname.
When verification fails, the counterexamples from failed VCs provide feedback for user-assisted error debugging.
Thanks to the modularity of \sysname, each VC violation is \emph{localized}, i.e., it is caused by the configurations or specifications of a certain node or edge. 
For example, when $\vcedge$ fails, the SMT solver will generate concrete routes $s_u, s_v$ such that $s_u\in Q(u)\land s_v\in I(v)$ but $s_v\oplus f_{(u,v)}(s_u) = s'_v \notin Q(v)$. Based on their beliefs (or knowledge) about $s_u, s_v, s_v'$, a user can strengthen $Q(u), I(v)$ to exclude $s_u, s_v$, or weaken $Q(v)$ to include $s'_v$.
In Appendix~\ref{appendix-subsec:debug} and Appendix~\ref{appendix-subsec:debug-example}, we describe a user-assisted debugging process with all cases and illustrate it with an example.
It is also possible that the transfer function $f_{(u,v)}$ is erroneous due to a bug in the network configuration and needs repair --- we leave automated repair to future work.

\section{Synthesis of Component Specifications}

\label{sec:synth}

In this section, we consider a different variation of the verification problem, where a user may have an idea of a \cbgraph, but would like to automatically synthesize the \interfaces $I, Q$, in order to prove the property $Y$ of interest. 
This is often the case in practice for highly structured networks such as those found in data centers: for instance, fat-tree topologies~\cite{al2008scalable} are configured to use routes along well-defined valley-free paths.
In such structured topologies and policies, a \cbgraph often corresponds to a BFS (breadth-first search) graph from the destination. 

Given a \cbgraph, our approach for synthesis of the \interfaces is based on 
Constrained Horn Clauses (CHC). 
CHC formulations are widely used in program verification and invariant synthesis~\cite{ChcPldi12,spacer,ChcFlc15}. We follow their notation in definitions shown below.
\begin{definition}
A constrained Horn clause (CHC) rule has one of the following two forms:
\begin{gather*}
    \forall x_1,\cdots,x_n.\ \phi(\vec x)\land R_1(\vec x)\land\cdots R_n(\vec x)\Rightarrow R_{n+1}(\vec x) \\
    \forall x_1,\cdots,x_n.\ \phi(\vec x)\land R_1(\vec x)\land\cdots R_n(\vec x)\Rightarrow\mathrm{False}
\end{gather*}
where $\phi$ is an interpreted formula (over some first-order theory) and each $R_i$ is an uninterpreted predicate. Note that the variables in each CHC rule ($x_1,\cdots,x_n$) are universally quantified --- these quantifiers are implicit when not shown.
\end{definition}

\begin{figure*}[t]
    \centering
    \begin{tabular}{rl}
        Interpreted: & $Y_v:R\to\mathbb B$, $\mathrm{init}_v\in R$, $f_{(u,v)}:R\to R$ \\
        Uninterpreted: & $I_v,Q_v:R\to\mathbb B$ \\
        \multicolumn{2}{c}{
            \begin{tabular}{cll}
                \toprule
                VC & For each & CHC rule \\
                \midrule
                $\vcinit(v)$ & $v\in V$ & $s_v=\mathrm{init}_v\Rightarrow I_v(s_v)$ \\
                \midrule
                $\vcprop(v)$ & $v\in V$ & $\neg Y_v(s_v)\land Q_v(s_v)\Rightarrow\mathrm{False}$ \\
                \midrule
                $\vcinv((u,v))$ & $(u,v)\in E$ & $s_v'=s_v\oplus f_{(u,v)}(s_u)\land I_u(s_u)\land I_v(s_v)\Rightarrow I_v(s_v')$ \\
                \midrule
                \multirowcell{2}{$\vcroot$(v)} & $v\in\mathrm{CBRoot}$ & $s_v=\mathrm{init}_v\Rightarrow Q_v(s_v)$ \\
                \cmidrule{2-3}
                & $v\in\mathrm{CBRoot}$ and $(u,v)\in E$
                & $s_v'=s_v\oplus f_{(u,v)}(s_u)\land I_u(s_u)\land Q_v(s_v)\Rightarrow Q_v(s_v')$ \\
                \midrule
                $\vcedge(u,v)$ & $(u,v)\in\mathrm{CBEdge}$ & $s_v'=s_v\oplus f_{(u,v)}(s_u)\land Q_u(s_u)\land I_v(s_v)\Rightarrow Q_v(s_v')$ \\
                \bottomrule
            \end{tabular}
        }
    \end{tabular}
    \caption{CHC system for the network component specification synthesis problem, given the \cbgraph. $\mathbb B=\{\mathrm{True},\mathrm{False}\}$ is the Boolean domain, variables $s_u,s_v,s_v'$ in CHC rules are universally quantified implicitly.}
    \label{fig:chc}
\end{figure*}

A CHC system is a conjunction of CHC rules. A solution of a CHC system is an interpretation of predicates $R_i$ (typically in the form of formulas in some first-order theory) such that all CHC rules are valid. 
A CHC solver (such as Spacer/Z3~\cite{spacer-2016}) searches for a solution; a CHC system is called infeasible if no such solution exists. 

\para{Formulating \interface synthesis in CHC}
We formulate CHC rules based on the VCs we showed earlier, but where $I_v, Q_v$ are now uninterpreted predicates to be solved, rather than interpreted predicates provided by a user. 
The \cbgraph ($\mathrm{CBRoot}\subseteq V$, $\mathrm{CBEdge}\subseteq E$) is provided by the user, property $Y_v$ and network logic ($\mathrm{init}_v$, $f_{(u,v)}$) are interpreted (based on the network configurations). 

The complete CHC system for the network \interface synthesis problem is shown in Fig.~\ref{fig:chc}. 
For each VC set in Fig.~\ref{fig:VCs}, we formulate a CHC rule for each member of the set (a node, edge, \cbroot or \cbedge).
Note that $\vcroot$ results in two CHC rules to conform to the definition of CHCs. 

\para{Prototype tool \synthname}
We have implemented \synthname based on the above CHC formulation.
Our tool uses BFS traversal to check that the user-provided \cbgraph is connected.
If a CHC solution exists, it corresponds to an interpretation of the $I_v, Q_v$ predicates that make all CHC rules valid, which establishes successful verification due to all VCs passing. 
An off-the-shelf CHC solver typically works with all CHC rules together in a non-modular way; we leave modular synthesis to future work.

\section{Implementation and Evaluation}
\label{sec:eval}

We have developed a prototype implementation of \sysname,
using Batfish~\cite{batfish} to extract the network model from Cisco and Juniper configurations.
The tool is implemented in C\#, using Microsoft's Zen library~\cite{zenlib} and SMT solver Z3~\cite{z3}. 
The \synthname prototype translates the CHC rules (Fig.~\ref{fig:chc}) into SMTLib format~\cite{smtlib} and uses Spacer~\cite{spacer,spacer-2016} (version 4.13.0, part of Z3) to solve the CHC problem.

Since Spacer does not support the theory of sequences (needed for modeling $aspath$), we used a slightly simplified BGP model for the \synthname evaluations.
In this simplified model, we replace the field of AS-paths with an abstraction (its length) and encode communities as a boolean array (assuming the communities are among a few pre-determined values).
While \synthname is less expressive than \sysname, our evaluations are not affected since the transfer function and property involve only simplified fields in all benchmark instances.

\para{Evaluation Goals}
We seek to answer the following questions.
\begin{itemize}
\itemsep0em
\item \textbf{RQ1: Expressiveness.} Is \sysname able to verify a variety of properties for a diverse collection of networks, including data center fat-trees~\cite{al2008scalable} as well as other real-world topologies and policies?
\item \textbf{RQ2: Performance of verification.} 
How well does \sysname perform as a verifier? 
Does it scale on large networks and complex configurations?
\item \textbf{RQ3: Usability of fault tolerance analysis.} Is \sysname able to perform a push-button fault tolerance analysis based on \cbgraph? What information can users gain from this analysis?
\item \textbf{RQ4: Performance of \interface synthesis.} How well does \synthname perform in practice?
\end{itemize}

\para{Comparison of \sysname with other tools}
For a baseline comparison, we have implemented a monolithic SMT-based verifier (MS) based on the algorithm in Minesweeper~\cite{minesweeper}.
We also compare \sysname with Timepiece~\cite{timepiece}, using an artifact~\cite{timepiece,timepiece-artifact} (with their user-provided \interfaces). 
Lightyear tool~\cite{lightyear} is not available (confirmed by the authors).
We note that since the inputs (properties and annotations) of \sysname are different from those in  other tools (MS, Timepiece), and they require different levels of user effort, it is not possible to perform a fair apples-to-apples comparison.
However, our goal is to focus on the same properties (such as reachability) that are of interest to network operators.
We report on the success and runtimes of different tools to provide information about their efficacy and practicality, which can be viewed as more indicative of their capabilities rather than a quantitative assessment.

\subsection{Benchmarks and Setup}
We worked with three sets of benchmarks, described below.\footnote{All benchmarks are available at \url{https://github.com/dz7903/cbgraphs-benchmarks}. An artifact of our prototype and benchmarks is available at \url{https://doi.org/10.5281/zenodo.22103905}.}

\para{(1) Synthetic fattrees} We use parameterized fattree networks (20--2000 nodes) with auto-generated configurations in Cisco's format (via Batfish frontend~\cite{batfish})
(with about 50-300 lines of code per node).
We verify the following four properties often used in datacenters:
\begin{itemize}
    \itemsep0em
    \item \bench{Reachability}: all nodes must have reachability to a fixed destination;
    \item \bench{PathLength}: the selected route at all nodes (to a fixed destination) must be the shortest path;
    \item \bench{ValleyFree}: the selected route at all nodes (to a fixed destination) must not be a valley (up-down-up);
    \item \bench{Hijack}: certain routes from an external "hijacker" node must be blocked from the internal network.
\end{itemize}
We remark that Timepiece~\cite{timepiece} used the same properties, and Lightyear does not support the \bench{PathLength} and \bench{ValleyFree} properties.
Additional details with full specifications used in \sysname are provided in Appendix~\ref{appendix-subsec:benchmark}.

\para{(2) Internet2} 
The Internet2 benchmark~\cite{internet2} is a real wide-area network with over 100,000 lines of Juniper configurations. It has 293 nodes (10 internal routers, 283 external neighbors).\footnote{
This is different from the reported number of 253 external neighbors in Timepiece's artifact~\cite{timepiece-artifact}, due to different logic in the frontend 
for parsing the configurations.}
We use a snapshot of Internet2's configurations~\cite{bagpipe} that have been reported to violate some properties.
\footnote{Some features used by Internet2, \EG IPv6, next-hop self, are not modeled. These do not affect our evaluation.}
We consider the following three properties:
\begin{itemize}
    \itemsep0em
    \item \bench{BlockToExternal} property (presented in Bagpipe~\cite{bagpipe}): check that internal routes with a BTE community are not advertised to Internet2's peers. 
    \item \bench{NoMartian} property (presented in Bagpipe~\cite{bagpipe}): check that internal routers should not select routes with external Martian prefixes as destination.
    \item \bench{InternalReachability} property: check that internal nodes can advertise themselves to other internal nodes.
\end{itemize}

\para{Batfish tutorial networks} We consider two illustrative examples from the Batfish tutorials~\cite{batfish-tutorials-bgp,batfish-tutorials-failure}. The first (for a school) has 13 routers running BGP. The second (shown in Fig.~\ref{fig:fault-tolerance-network}) has 13 routers distributed across Asia, Europe, and the US. 
We verify reachability and AS path length (i.e., AS paths are of length 2 or less) properties in the first example, and reachability to Houston in the second example.

\para{Experimental setup} 
All experiments were run on a Linux server with 16 Intel(R) Xeon(R) CPUs and 252 GB memory. We ran the \sysname and Timepiece modular verifiers in parallel using all 16 CPUs, and run the monolithic verifier (which is not easily parallelizable) on 1 CPU.

\subsection{Evaluation Results for \sysname}

\begin{figure}[t]
    \centering
    \begin{tikzpicture}
        \begin{axis}[
            scale=0.22,
            at={(0,0)},
            title={(a) \bench{Reachability}},
            title style={yshift=-3.5cm},
            xlabel={\#Nodes},
            ylabel={Verification time (s)},
            xmin=0, xmax=2000,
            ymin=0.1, ymax=86400,
            xtick={0,1000,2000},
            ymode=log,
            grid style=dashed,
        ]
        \addplot[color=blue,mark=square]
            coordinates {
            (20,0.348)(80,0.969)(180,2.822)(320,7.647)(500,19.269)(720,45.365)(980,100.438)(1280,197.081)(1620,389.699)(2000,774.842)
            };
        \addplot[color=green,mark=square]
            coordinates {
            (20,0.203)(80,0.292)(180,0.591)(320,1.328)(500,2.345)(720,3.888)(980,6.080)(1280,8.964)(1620,12.856)(2000,20.175)
            };
        \addplot[color=red, only marks, mark=o]
            coordinates {
            (20,0.329)(80,4.813)(180,38.808)(320,710.890)(500,4872.436)(720,9541.630)(980,34585.131)
            };
        \addplot[color=red, only marks, mark=x]
            coordinates {
            (1280,86400)(1620,86400)(2000,86400)
            };
        \addplot[color=red]
            coordinates {
            (20,0.329)(80,4.813)(180,38.808)(320,710.890)(500,4872.436)(720,9541.630)(980,34585.131)
            (1280,86400)(1620,86400)(2000,86400)
            };
        \end{axis}
        \begin{axis}[
            scale=0.22,
            at={(4cm,0)},
            title={(b) \bench{PathLength}},
            title style={yshift=-3.5cm},
            xlabel={\#Nodes},
            ylabel={Verification time (s)},
            xmin=0, xmax=2000,
            ymin=0.1, ymax=86400,
            xtick={0,1000,2000},
            ymode=log,
            grid style=dashed,
            legend style={nodes={scale=0.9, transform shape},at={(1.1,2)},legend columns=-1},
        ]
        \addplot[color=blue,mark=square]
            coordinates {
            (20,0.329)(80,1.203)(180,3.605)(320,8.721)(500,21.241)(720,47.207)(980,103.592)(1280,199.915)(1620,365.722)(2000,585.614)
            };
        \addplot[color=green,mark=square]
            coordinates {
            (20,0.228)(80,0.712)(180,3.021)(320,9.340)(500,24.679)(720,45.879)(980,111.903)(1280,183.326)(1620,315.450)(2000,497.188)
            };
        \addplot[color=red,mark=o]
            coordinates {
            (20,0.440)(80,11.395)(180,189.179)(320,1318.315)(500,6652.372)(720,21846.499)
            };
        \addplot[color=red,mark=x]
            coordinates {
            (980,86400)(1280,86400)(1620,86400)(2000,86400)
            };
        \addplot[color=red]
            coordinates {
            (20,0.440)(80,11.395)(180,189.179)(320,1318.315)(500,6652.372)(720,21846.499)
            (980,86400)(1280,86400)(1620,86400)(2000,86400)
            };
        \legend{\sysname,Timepiece,MS,Timeout}
        \end{axis}
        \begin{axis}[
            scale=0.22,
            at={(0,-4cm)},
            title={(c) \bench{ValleyFree}},
            title style={yshift=-3.5cm},
            xlabel={\#Nodes},
            ylabel={Verification time (s)},
            xmin=0, xmax=2000,
            ymin=0.1, ymax=86400,
            xtick={0,1000,2000},
            ymode=log,
            grid style=dashed,
        ]
        \addplot[color=blue,mark=square]
            coordinates {
            (20,0.478)(80,2.069)(180,6.023)(320,13.779)(500,25.537)(720,47.092)(980,73.308)(1280,120.509)(1620,207.695)(2000,393.430)
            };
        \addplot[color=green,mark=square]
            coordinates {
            (20,0.233)(80,0.987)(180,5.518)(320,19.764)(500,53.976)(720,122.048)(980,283.635)(1280,474.987)(1620,754.306)(2000,1165.468)
            };
        \addplot[color=red,only marks,mark=o]
            coordinates {
            (20,0.720)(80,342.200)(180,1309.208)(320,17657.073)
            };
        \addplot[color=red,only marks,mark=x]
            coordinates {
            (500,86400)(720,86400)(980,86400)(1280,86400)(1620,86400)(2000,86400)
            };
        \addplot[color=red]
            coordinates {
            (20,0.720)(80,342.200)(180,1309.208)(320,17657.073)
            (500,86400)(720,86400)(980,86400)(1280,86400)(1620,86400)(2000,86400)
            };
        \end{axis}
        
        \begin{axis}[
            scale=0.22,
            at={(4cm,-4cm)},
            title={(d) \bench{Hijack}},
            title style={yshift=-3.5cm},
            xlabel={\#Nodes},
            ylabel={Verification time (s)},
            xmin=0, xmax=2000,
            ymin=0.1, ymax=86400,
            xtick={0,1000,2000},
            ymode=log,
            grid style=dashed,
        ]
        \addplot[color=blue,mark=square]
            coordinates {
            (20,0.534)(80,2.479)(180,7.033)(320,15.863)(500,30.053)(720,54.145)(980,84.860)(1280,123.536)(1620,176.160)(2000,246.162)
            };
        \addplot[color=green,mark=square]
            coordinates {
            (20,0.217)(80,0.828)(180,2.200)(320,7.386)(500,13.301)(720,33.194)(980,47.209)(1280,74.829)(1620,110.170)(2000,171.972)
            };
        \addplot[color=red,mark=o]
            coordinates {
            (20,23.694)(80,1577.377)
            };
        \addplot[color=red,mark=x]
            coordinates {
            (180,86400)(320,86400)(500,86400)(720,86400)(980,86400)(1280,86400)(1620,86400)(2000,86400)
            };
        \addplot[color=red]
            coordinates {
            (80,1577.377)(180,86400)
            };
        \end{axis}
    \end{tikzpicture}
    \caption{Evaluation results for \sysname on fattree networks, compared with Timepiece (a modular verifier) and MS (a Minesweeper-style monolithic verifier).}
    \label{fig:verification-performance-fattree}
\end{figure}

\para{Fattree networks} 
\label{subsec:eval-verif}
The evaluation results (Fig.~\ref{fig:verification-performance-fattree}) show the performance of \sysname, Timepiece, and the baseline monolithic verifier (MS).
Each graph plots the number of nodes in the network ($x$-axis) against the wall-clock time (in seconds, in log scale on the $y$-axis), with a timeout of 24 hours.

Note first that \sysname successfully verifies all properties on all networks. This answers $RQ1$ positively: \sysname is able to express and verify these properties deemed useful by network operators.
Regarding $RQ2$, these graphs show that \sysname is able to verify all properties on a 2000-node fattree networks within 20 minutes, whereas the monolithic verifier (MS) has many timeouts, e.g., it fails to verify the \bench{Hijack} property on a 180-node fattree.
The performance of \sysname is generally comparable with Timepiece, marginally better for the \bench{ValleyFree} benchmark, but worse in the \bench{Reachability} benchmark (where the property is simple and the \cbgraph checks in \sysname are an overhead). 

\para{Internet2}
The evaluation results of \sysname on Internet2 network are shown in Fig.~\ref{fig:eval-internet2}, compared with Timepiece \footnote{The Timepiece artifact~\cite{timepiece-artifact} does not include \bench{\footnotesize NoMartians} and \bench{\footnotesize InternalReachability} benchmarks, so we evaluate the version published by Alberdingk Thijm in GitHub~\cite{timepiece-github}.} and MS (Minesweeper).
\sysname reports 3 counterexamples (cex) for \bench{BlockToExternal} and 1 counterexample for \bench{NoMartians}.
We manually inspected the configurations and found that the reported counterexamples violate the invariants according to our BGP model.
However, we note that these may not be actual bugs under a real-world deployment with additional environment constraints.
Timepiece does not report these violations (since it excludes certain locations in its specifications).
It completes verification successfully, taking much longer for \bench{InternalReachability}.
MS is unable to complete any benchmark on Internet2.

\begin{figure}[t]
    \centering
    \small{
    \begin{tabular}{cccc}
        \toprule
        Property & \sysname & Timepiece & MS \\
        \midrule
        \bench{BlockToExternal}
        & 3 cex, 22s
        & 42.4m
        & OoM \\
        \midrule
        \bench{NoMartians}
        & 1 cex, 121s
        & 11.1m
        & OoM \\
        \midrule
        \bench{InternalReachability}
        & 74s
        & 11.4m
        & TO \\
        \bottomrule
    \end{tabular}
    \caption{Evaluation results on Internet2 network.
    ``s'' is seconds, ``m'' is minutes, cex indicates counterexample, OoM indicates Out-of-Memory, TO indicates timeout after 24 hours.}
    \label{fig:eval-internet2}
    }
\end{figure}

\para{Batfish tutorial networks}
\sysname takes 0.4s to verify each of two properties (reachability, AS path length) in the school network.
It takes 0.5s to verify the reachability property in the cross-world network. Since MS has limitations in handling IGP, it gave a counterexample (which is spurious).

\para{Fault tolerance analysis}
To answer RQ3, we run the fault tolerance analysis on each benchmark that passes the validity checks.
For the largest 2000-node fattree network, \sysname takes less than 20 seconds to check \cbgraph connectivity in each benchmark. The network can tolerate 19-edge failures for the properties between top-of-rack routers.

In the Internet2 benchmark, \sysname takes 23 ms (milliseconds) for fault tolerance analysis in \bench{InternalReachability}, and finds the network tolerant to 8-edge failures. In the Batfish tutorial benchmarks, \sysname finds that the school network can tolerate 1-edge failure for reachability, but does not tolerate any edge failure for the AS path length property.
In the second (cross-world) network, \sysname finds that the routers in Asia do not tolerate any edge failure. (as discussed in \S\ref{subsec:fault-tolerance}). Each failure analysis takes 8 ms.

\begin{figure}
    \centering
    \small{
    \begin{tabular}{ccc}
        Benchmark & Lines of Code (in C\#) \\
        \toprule
        \bench{Reachability} & 22 (12 for common definitions) \\
        \midrule
        \bench{PathLength} & 32 (12 for common definitions) \\
        \midrule
        \bench{ValleyFree} & 48 (13 for common definitions) \\
        \midrule
        \bench{Hijack} & 44 (21 for common definitions) \\
        \midrule
        \bench{BlockToExternal} & 49 (18 for common definitions) \\
        \midrule
        \bench{NoMartians} & 61 (30 for common definitions) \\
        \midrule
        \bench{InternalReachability} & 37 (17 for common definitions) \\
        \bottomrule
    \end{tabular}
    }
    \caption{User effort in writing specifications in \sysname.}
    \label{fig:loc}
\end{figure}

\para{User effort}
Figure~\ref{fig:loc} reports the user effort in providing specifications in \sysname, in terms of lines of code (LoC) in C\# (via the Zen library~\cite{zenlib}), with many common definitions (as shown). These include properties $Y$ and annotations $I, Q$ for all nodes and their initial routes. (The tool library provides the BGP model and transfer function encoding into SMT.)
Different sizes of the fattree networks (20 nodes to 2000 nodes) require the same user effort in LoC, because specifications in our prototype can be written as parameterized C\# functions that are reused for different sizes.

In terms of effort for refining the annotations, we needed 2-3 iterations to find the correct annotations ($I, Q$) for \bench{ValleyFree}, \bench{Hijack} and \bench{BlockToExternal} benchmarks. We did not need any refinement for \bench{Reachability}, \bench{PathLength}, \bench{NoMartians}, and \bench{InternalReachability} benchmarks since their annotations are relatively simple.

\subsection{Results for Synthesis of Component Specifications}

To answer \textbf{RQ4}, we perform experiments on the synthetic fattree benchmarks, with the four properties described earlier. The \cbgraph we provided is in BFS order from the destination. 
Fig.~\ref{fig:synth} shows the evaluation results, with each graph plotting the number of nodes ($x$-axis) and the wall-clock time in seconds (in log scale, $y$-axis) with a timeout of 1 hour.
We compare \synthname with MS (Minesweeper) because it does not require user-provided \interfaces to perform verification for the whole network.
(The performance of MS here is better than in Fig.~\ref{fig:verification-performance-fattree} due to a simplified BGP model, as mentioned earlier.)

In almost all experiments (except one), \synthname is faster than monolithic verification in MS. 
Among improvements, \synthname is 22-27x faster for checking \bench{Reachability} (best case), and is 1.3x faster for checking \bench{ValleyFree} (least improved).
(The fluctuations in \bench{PathLength} and \bench{Hijack} were reproduced multiple times in our experiments, likely due to solver heuristics.)
These results are encouraging for \interface synthesis: they show that even with an out-of-the-box CHC solver (not necessarily modular), \synthname shows an improvement against monolithic verification.

\begin{figure}[t]
    \centering
    \begin{tikzpicture}
        \begin{axis}[
            scale=0.22,
            at={(0,0)},
            title={(a) \bench{Reachability}},
            title style={yshift=-3.2cm},
            xlabel={\#Nodes},
            ylabel={Verification time (s)},
            xmin=0, xmax=2000,
            ymin=0.01, ymax=1000,
            xtick={0,1000,2000},
            ymode=log,
            grid style=dashed,
        ]
        \addplot[color=blue,mark=square]
            coordinates {
            (20,0.034)(80,0.051)(180,0.092)(320,0.182)(500,0.333)(720,0.596)(980,0.868)(1280,1.485)(1620,1.937)(2000,2.907)
            };
        \addplot[color=red, mark=o]
            coordinates {
            (20,0.162)(80,0.659)(180,2.261)(320,4.543)(500,8.649)(720,14.664)(980,23.932)(1280,33.411)(1620,47.388)(2000,65.810)
            };
        \end{axis}
        \begin{axis}[
            scale=0.22,
            at={(3cm,0)},
            title={(b) \bench{PathLength}},
            title style={yshift=-3.2cm},
            xlabel={\#Nodes},
            yticklabel=\empty,
            xmin=0, xmax=2000,
            ymin=0.01, ymax=3600,
            xtick={0,1000,2000},
            ymode=log,
            grid style=dashed,
            legend style={nodes={scale=0.9, transform shape},at={(1,1.5)},legend columns=-1},
        ]
        \addplot[color=blue,mark=square]
            coordinates {
            (20,0.050)(80,0.273)(180,0.815)(320,12.422)(500,84.474)(720,280.468)(980,37.642)(1280,2772.808)
            };
        \addplot[color=red, mark=o]
            coordinates {
            (20,0.364)(80,3.766)(180,19.774)(320,90.459)(500,297.609)(720,793.607)(980,2756.459)
            };
        \addplot[color=red, mark=x]
            coordinates {
            (1280,3600)(1620,3600)(2000,3600)
            };
        \addplot[color=blue,mark=x]
            coordinates {
            (1620,3600)(2000,3600)
            };
        \addplot[color=blue]
            coordinates {
            (1280,2772.808)(1620,3600)
            };
        \addplot[color=red]
            coordinates {
            (980,2756.459)(1280,3600)
            };
        \legend{\synthname,MS,Timeout}
        \end{axis}
        \begin{axis}[
            scale=0.22,
            at={(0,-3.2cm)},
            title={(c) \bench{ValleyFree}},
            title style={yshift=-3.2cm},
            xlabel={\#Nodes},
            ylabel={Verification time (s)},
            xmin=0, xmax=2000,
            ymin=0.01, ymax=3600,
            xtick={0,1000,2000},
            ymode=log,
            grid style=dashed,
        ]
        \addplot[color=blue,mark=square]
            coordinates {
            (20,0.060)(80,1.186)(180,14.503)(320,217.973)(500,859.388)
            };
        \addplot[color=blue,mark=x]
            coordinates {
            (720,3600)(980,3600)(1280,3600)(1620,3600)(2000,3600)
            };
        \addplot[color=red, mark=o]
            coordinates {
            (20,0.611)(80,12.747)(180,88.341)(320, 367.864)(500,1153.007)(720,3056.290)
            };
        \addplot[color=red,mark=x]
            coordinates {
            (980,3600)(1280,3600)(1620,3600)(2000,3600)
            };
        \addplot[color=blue]
            coordinates {
            (500,859.388)(720,3600)
            };
        \addplot[color=red]
            coordinates {
            (720,3056.290)(980,3600)
            };
        \end{axis}
        \begin{axis}[
            scale=0.22,
            at={(3cm,-3.2cm)},
            title={(d) \bench{Hijack}},
            title style={yshift=-3.2cm},
            xlabel={\#Nodes},
            yticklabel=\empty,
            xmin=0, xmax=2000,
            ymin=0.01, ymax=3600,
            xtick={0,1000,2000},
            ymode=log,
            grid style=dashed,
        ]
        \addplot[color=blue,mark=square]
            coordinates {
            (20,0.067)(80,0.166)(180,1.215)(320,2.413)(500,10.298)(720,59.387)(980,59.116)(1280,120.520)(1620,1483.435)(2000,141.669)
            };
        \addplot[color=red, mark=o]
            coordinates {
            (20,0.187)(80,1.435)(180,8.819)(320,72.854)(500,269.397)(720,1030.978)(980,2491.323)
            };
        \addplot[color=red,mark=x]
            coordinates {
            (1280,3600)(1620,3600)(2000,3600)
            };
        \addplot[color=red]
            coordinates {
            (980,2491.323)(1280,3600)
            };
        \end{axis}
    \end{tikzpicture}

    \caption{Results for automated \interface synthesis (\synthname), compared with a monolithic verifier (MS).}
    \label{fig:synth}
\end{figure}
\section{Related Work}
\label{sec:related}

\para{Modular Network Verification}
\sysname improves upon related modular verifiers~\cite{rcdc,kirigami,timepiece,lightyear} in various aspects. 
RCDC~\cite{rcdc} was a pioneer in modular network verification, but was customized for checking invariants of Azure data centers, not for other settings. 
Kirigami~\cite{kirigami} requires a user to exactly identify the concrete converged state of a network, a difficult task requiring exact knowledge of network internals.
Lightyear~\cite{lightyear} is able to verify safety and liveness properties, but not those that involve path lengths or path preferences. It requires users to provide invariants (for safety) and a specific path and path constraints (for liveness), thereby placing a higher burden on users than in \sysname, which automatically generates a \cbgraph.
Timepiece~\cite{timepiece} requires the user to provide component specifications as functions that describe how each component’s routing state changes over all time instants. While it can support more expressive properties (also expressed as functions over time) than \sysname, this comes at the cost of higher user burden. Furthermore, Timepiece considers only a synchronous schedule, whereas \sysname is proven sound under asynchronous network semantics.

\para{Control plane verification}
Our work is more broadly related to verification of network control plane configurations~\cite{bagpipe,minesweeper,arc,abhashkumar2020tiramisu,prabhu2020plankton,ye2020accuracy,beckett2018control,beckett2019abstract,acorn,expresso}.
Bagpipe~\cite{bagpipe} and Minesweeper~\cite{minesweeper} use SMT-based verification, but are monolithic verifiers.
ARC~\cite{arc} and Tiramisu~\cite{abhashkumar2020tiramisu} reduce verification problems to efficient graph analyses, but do not handle all network features.
Plankton~\cite{prabhu2020plankton}, Hoyan~\cite{ye2020accuracy}, and Expresso~\cite{expresso} combine a mix of symbolic analyses and simulation to improve scalability.
Bonsai~\cite{beckett2018control} uses symmetry-based abstractions, Shapeshifter~\cite{beckett2019abstract} uses abstractions on routing fields, and ACORN~\cite{acorn} uses route nondeterminism to improve scalability --- these abstractions are orthogonal to our work.

\para{Reasoning about convergence} Other efforts focus on network routing protocols, to check convergence and identify conditions under which (unique) convergence is guaranteed~\cite{griffin2005metarouting,sobrinho2005algebraic,daggitt2018asynchronous}. However, practical networks often use rich routing policies that do not conform to these conditions. Furthermore, these efforts do not target checking network- and policy-specific routing path properties for given network instances.

\para{Program verification and invariant synthesis} CHC solvers have been widely used in program verification and invariant synthesis~\cite{ChcPldi12,spacer,ChcFlc15}.
In addition to invariants $I$, we use specifications $Q$ to prove properties that are eventually-stable. 

\section{Conclusions}

We developed \sysname, a new framework for verifying eventually-stable control plane properties in a modular way, based on abstracting convergence and automatically generating a \emph{converges-before (CB) graph} from given \interfaces.
We proved our verifier correct in Lean, and illustrate its  benefits in fault tolerance analysis. 
We evaluated \sysname on a collection of benchmark examples, illustrating its expressiveness and performance. 
Given a \cbgraph, we also performed automated \interface synthesis and show its effectiveness in practice.

\section{Acknowledgment and Disclosure}
\para{Acknowledgments}
We thank the anonymous reviewers for their feedback and suggestions. We thank Mia Kaarls for development of parts of the frontend and backend of \sysname. This work was supported (in part) by NSF grants CNS-2312539 and CNS-2319442.

\para{AI disclosure}
AI was used to proofread the paper, correct grammatical errors, and improve formatting. AI was also used in developing an artifact for our prototype.

\bibliographystyle{plain}
\bibliography{refs,refs-ag}

\appendix

\subsection{Additional Resources}

\label{appendix-subsec:repo}

\begin{itemize}
    \item 
A Lean formalization of the network semantics and proofs is available at \url{https://github.com/dz7903/cbgraphs-formalization}.

    \item 
The benchmarks used for evaluations are available at \url{https://github.com/dz7903/cbgraphs-benchmarks}.

    \item 
An artifact of our prototype and benchmarks is available at \url{https://doi.org/10.5281/zenodo.22103905}~\cite{cbver-artifact}. Please see the \texttt{README.md} on running the benchmarks. Note that this artifact includes \sysname and our Batfish based frontend, but does not include our Lean formalization and the compared Timepiece verifier (the Timepiece artifact and source code are available at \url{https://doi.org/10.5281/zenodo.7799158}~\cite{timepiece-artifact} and \url{https://github.com/NetworkVerification/Timepiece/tree/main}~\cite{timepiece-github}).

\end{itemize} 
\subsection{Key definitions and theorems}

\label{appendix-subsec:proof}

We provide a catalog of key definitions and theorems used in \S\ref{sec:theory} and \S\ref{sec:tool} in this section, and connect them with Lean definitions and proofs, which are presented in the repository in Appendix~\ref{appendix-subsec:repo}. Every theorem we present in these sections has a formalized proof in Lean.

\subsubsection{Key definitions}

We quickly review the definitions of a network instance $N$, an asynchronous schedule $S$, and the network semantics with respect to $N$ and $S$, which have been presented in \S\ref{sec:theory}.

\begin{definition}[Network instance]
    A \emph{network} is a 5-tuple $N=(G,R,\mathrm{init},f,\oplus)$ where
    \begin{itemize}
        \itemsep0em
        \item $G=(V,E)$ is the directed graph with nodes $V$ and edges $E$.
        \item $R$ is the set of routes. We assume a least-preferred element $\infty\in R$.
        \item $\mathrm{init}:V\to S$ is the initial routes for each node.
        \item $f:E\to(S \to S)$ are transfer functions. We may write $f(e)$ as $f_e$.
        \item $\oplus: S\times S\to S$ is a binary associative, commutative and selective function (the \emph{merge} function).
    \end{itemize}

    This definition is formalized as \codelink{Network}{https://github.com/dz7903/cbgraphs-formalization/blob/322519adb499fbf4ec8516ba46523e710cf8b897/FormalCbgraphs/Network.lean\#L5} in Lean.
\end{definition}
Note that in routing algebras, $\infty$ and $\oplus$ are usually defined without a partial order over $R$ (and the partial order is usually derived as $a\le b\Leftrightarrow a\oplus b=a$).
In our Lean defintion we use type classes \codelink{OrderTop}{https://leanprover-community.github.io/mathlib4_docs/Mathlib/Order/BoundedOrder/Basic.html\#OrderTop} and \codelink{SemilatticeInf}{https://leanprover-community.github.io/mathlib4_docs/Mathlib/Order/Lattice.html\#SemilatticeInf} in Mathlib (instead of defining them on our own), which assume both a partial order and the algebraic operations.

\begin{definition}[Asynchronous schedule]
    A \emph{schedule} of a network $N$ is a pair $(\alpha,\beta)$ where
    \begin{itemize}
        \itemsep0em
        \item $\alpha:\mathbb T\to 2^V$ is the activation function.
        \item $\beta:E\to(\mathbb T\to\mathbb T)$ is the data flow function (we also write $\beta(e)$ as $\beta_e$).
        \item $\beta$ satisfies the \emph{causality axiom} that messages are delivered after sent, \IE $$\forall e.\ \forall t>0.\ \beta_e(t)<t.$$
    \end{itemize}

    This definition is formalized as \codelink{Schedule}{https://github.com/dz7903/cbgraphs-formalization/blob/322519adb499fbf4ec8516ba46523e710cf8b897/FormalCbgraphs/Network.lean\#L12} in Lean.
\end{definition}

\begin{definition}[Network semantics]
    The \emph{network semantics} of a network $N$ respect to a schedule $S=(\alpha,\beta)$ is a function $\sigma_S:V\times\mathbb T\to S$ defined as
    \begin{equation}
    \sigma_S(v,t)=
    \begin{cases}
    I(v) & t=0 \\
    \sigma_S(v,t-1) & t>0\land v\notin\alpha(t) \\
    \multispan2{$\mathrm{init}(v)\oplus\bigoplus_uf_{(u,v)}(\sigma_S(u,\beta_{(u,v)}(t)))$\hfill} \\
    \phantom{\bigoplus_{(u,v)\in E}f_{u,v}(\sigma(u,\beta_{u,v}(t)))} & t>0\land v\in\alpha(t)
    \end{cases}
    \end{equation}

    This definition is formalized as \codelink{Schedule.sem}{https://github.com/dz7903/cbgraphs-formalization/blob/322519adb499fbf4ec8516ba46523e710cf8b897/FormalCbgraphs/Network.lean\#L47} in Lean.
\end{definition}

Now we present several conditions related to the fairness of a schedule.

\begin{definition}
    Let $S=(\alpha,\beta)$ be a schedule on a network $N$.
    \begin{itemize}
        \itemsep0em
        \item A node $v$ is said \emph{non-failed} if it is activated in $\alpha$ infinitely often, \IE $$\forall T.\ \exists t\ge T.\ v\in\alpha(t).$$
        \item An edge $e$ is said \emph{non-failed} if no message at a certain time is sent infinite times (messages up to that time must be flushed away since a certain time), \IE $$\forall T.\ \exists T'\ge T.\ \forall t\ge T'.\ \beta_e(t)\ge T.$$
    \end{itemize}

    These definitions are formalized as \codelink{Schedule.NodeNonFailed}{https://github.com/dz7903/cbgraphs-formalization/blob/322519adb499fbf4ec8516ba46523e710cf8b897/FormalCbgraphs/Network.lean\#L23} and \codelink{Schedule.EdgeNonFailed}{https://github.com/dz7903/cbgraphs-formalization/blob/322519adb499fbf4ec8516ba46523e710cf8b897/FormalCbgraphs/Network.lean\#L27} in Lean.
\end{definition}

All these conditions are reasonable since BGP-over-TCP guarantees messages are delivered eventually and in-order.

The formal definition of fairness is shown below.
\begin{definition}[Fairness]
    A schedule $S$ is \emph{fair} if every node and every edge is non-failed. $S$ is \emph{fair with at most $k$ failures} if every node is non-failed and there exists a set $F\subseteq E$ with $|F|\le k$ (the set of failed edges) such that all edges in $E-F$ are non-failed.

    These definitions are formalized as \codelink{Schedule.Fair}{https://github.com/dz7903/cbgraphs-formalization/blob/322519adb499fbf4ec8516ba46523e710cf8b897/FormalCbgraphs/Network.lean\#L31} and \codelink{Schedule.FairWithFailure}{https://github.com/dz7903/cbgraphs-formalization/blob/322519adb499fbf4ec8516ba46523e710cf8b897/FormalCbgraphs/Network.lean\#L36} in Lean.
\end{definition}

We also define the concept of connectedness and $k$-fail-connectivity:
\begin{definition}[Connectedness]
    In a graph $G=(V,E)$, let $Roots\subseteq V$ be a set of nodes and $Edges\subseteq E$ be a set of edges. We say a node $v$ is \emph{connected} under $(Roots,Edges)$ if there exists a path in $Edges$ from a node $u\in Roots$ to $v$.
    
    We say a node $v$ is $k$-fail-connected under $(Roots,Edges)$ if it is connected under $(Roots,Edges-F)$ for any $F\subseteq Edges$ with $|F|\le k$.

    These definitions are formalized as \codelink{Graph.Connected}{https://github.com/dz7903/cbgraphs-formalization/blob/322519adb499fbf4ec8516ba46523e710cf8b897/FormalCbgraphs/Graph.lean\#L17} and \codelink{Graph.ConnectedWithFailure}{https://github.com/dz7903/cbgraphs-formalization/blob/322519adb499fbf4ec8516ba46523e710cf8b897/FormalCbgraphs/Graph.lean\#L23} in Lean.
\end{definition}

We define \interfaces and verification conditions in bundle as follows.
\begin{definition}
    A \emph{verification condition structure} over a network $N$ and properties $Y:V\to 2^R$ includes:
    \begin{itemize}
        \itemsep0em
        \item \interfaces $(I,Q)$,
        \item and a \cbgraph $CB=(CBRoots,CBEdges)$ (in practice this \cbgraph is inferred by \sysname), such that
        \item the verification conditions in Fig.~\ref{fig:VCs} all hold, and
        \item the \cbgraph $CB$ is connected.
    \end{itemize}

    This definition is formalized as \codelink{VC}{https://github.com/dz7903/cbgraphs-formalization/blob/322519adb499fbf4ec8516ba46523e710cf8b897/FormalCbgraphs/Soundness.lean\#L24} in Lean.
\end{definition}

Finally, for a $Q$-interface in such a verification condition structure \ag{what does this mean?} \dz{This is the definition 8. It is a bundle of annotations (including Q) and the proof of VCs. It's mainly for convenience. I restated the experssion to refer to Q.}, we define the concept of abstract converge as follows. \ag{please avoid danglers like this} \dz{fixed}
\begin{definition}
    A node $v\in V$ abstractly converges to $Q(v)$ under schedule $S$ at time $\tau$ if $$\forall t\ge\tau.\ \sigma_S(v,t)\in Q(v).$$

    This definition is formalized as \codelink{VC.AbstractlyConverge}{https://github.com/dz7903/cbgraphs-formalization/blob/322519adb499fbf4ec8516ba46523e710cf8b897/FormalCbgraphs/Soundness.lean\#L61} in Lean.
\end{definition}

\subsubsection{Theorems for correctness}

We present formal version of Theorem~\ref{thm:connected-cbgraph}, Theorem~\ref{thm:correctness} and Theorem~\ref{thm:fail-connected-cbgraph}, with several key lemmas formalized in Lean. We also provide proof sketches for each one.

The first lemma shows $I(v)$ holds for any time $t$ and any schedule:
\begin{lemma}
    \label{lemma:invariance}
    Given verification conditions hold, for any schedule $S$, any time $t\in\mathbb T$ and node $v\in V$, $$\sigma_S(v,t)\in I(v).$$
\end{lemma}
\begin{proof}[Proof sketch]
    Induction on $t$ and do case analysis on $t=0$ (which uses $\vcinit$), $t>0\land v\notin\alpha(t)$ (inductive hypothesis) and $t>0\land v\in\alpha(t)$ ($\vcinv$ and inductive hypothesis).

    This lemma is formalized as \codelink{VC.invariance}{https://github.com/dz7903/cbgraphs-formalization/blob/322519adb499fbf4ec8516ba46523e710cf8b897/FormalCbgraphs/Soundness.lean\#L42} in Lean.
\end{proof}

The second lemma shows \cbroots satisfy their semantics definition:
\begin{lemma}
    \label{lemma:cbroots}
    Given verification conditions hold, for any schedule $S$ and any $v\in CBRoots$, it abstractly converges to $Q(v)$.
\end{lemma}
\begin{proof}[Proof sketch]
    Induction on $t$ and do case analysis on $t=0$ (which uses the first part of $\vcroot$), $t>0\land v\notin\alpha(t)$ (inductive hypothesis) and $t>0\land v\in\alpha(t)$ (the second part of $\vcroot$ and Lemma~\ref{lemma:invariance}).
    
    This lemma is formalized as \codelink{VC.cbroot}{https://github.com/dz7903/cbgraphs-formalization/blob/322519adb499fbf4ec8516ba46523e710cf8b897/FormalCbgraphs/Soundness.lean\#L65} in Lean.
\end{proof}

The third lemma shows \cbedges satisfy their semantics definition and therefore form the inductive steps of abstract convergence:
\begin{lemma}
    \label{lemma:cbedges}
    Given verification conditions hold, for any schedule $S$ and $e=(u,v)\in CBEdges$, if $v$ is eventually activated, $e$ is eventually flushed in $S$, and $u$ abstractly converges to $Q(u)$, then $v$ abstractly converges to $Q(v)$.
\end{lemma}
\begin{proof}[Proof sketch]
    Suppose $\sigma_S(u,t)\in Q(u)$ for any $t\ge t_1$. By the fairness conditions there is a time $t_2>t_1$ such that $v\in\alpha(t_2)$ and $\forall t\ge t_2.\ \beta_{(u,v)}(t)\ge t_1$. Now for any $t\ge t_2$ we show $\sigma_S(v,t)\in Q(v)$ by induction on $t$ and case analysis on $t>t_2\land v\notin\alpha(t)$ and $t\ge t_2\land v\in\alpha(t)$. The former is trivial. For the latter, note that
    \begin{align*}
        \sigma_S(v,t)&=\mathrm{init}(v)\oplus\left(\bigoplus_{(u',v)\in E}\sigma_S(u',\beta_{(u',v)}(t))\right) \\
        &=\sigma_S(u,\beta_{(u,v)}(t))
        &\}s_1 \\
        &\quad\oplus\left(\mathrm{init}(v)\oplus\bigoplus_{\substack{(u',v)\in E \\ u'\ne u}}\sigma_S(u',\beta_{(u',v)}(t))\right)
        &\left.\vphantom{\left(\mathrm{init}(v)\oplus\bigoplus_{\substack{(u',v)\in E \\ u'\ne u}}\sigma_S(u',\beta_{(u',v)}(t))\right)}\right\}s_2
    \end{align*}
    By $\vcinit$, $\vcinv$ and Lemma~\ref{lemma:invariance} we have $s_2\in I(v)$. Then by $\vcedge$ we have $\sigma_S(v,t)\in Q(v)$.

    This lemma is formalized as \codelink{VC.cbedge}{https://github.com/dz7903/cbgraphs-formalization/blob/322519adb499fbf4ec8516ba46523e710cf8b897/FormalCbgraphs/Soundness.lean\#L83} in Lean.
\end{proof}

We now formally state Theorem~\ref{thm:connected-cbgraph} and Theorem~\ref{thm:fail-connected-cbgraph}:
\begin{manualtheorem}{\ref{thm:connected-cbgraph}}
    Given verification conditions hold, for any fair schedule $S$ and any node $v\in V$, $v$ abstractly converge to $Q(v)$ at some time $\tau_v$.
\end{manualtheorem}
\begin{proof}
    By induction on the connectedness of \cbgraph, the base step is Lemma~\ref{lemma:cbroots} and the inductive step is Lemma~\ref{lemma:cbedges}.
    
    This theorem is formalized as \codelink{VC.connected\_cbgraph}{https://github.com/dz7903/cbgraphs-formalization/blob/322519adb499fbf4ec8516ba46523e710cf8b897/FormalCbgraphs/Soundness.lean\#L116} in Lean.
\end{proof}

\begin{manualtheorem}{\ref{thm:fail-connected-cbgraph}}
    Given verification conditions hold, for any schedule $S$ that is fair with at most $k$ failures, and any node $v\in V$, if $v$ is $k$-fail-connected under the \cbgraph $CB$, then $v$ abstractly converge to $Q(v)$.
\end{manualtheorem}
\begin{proof}
    By induction on the connectedness of \cbgraph excluding all failed edges in $S$. The remained is the same as proof of Theorem~\ref{thm:connected-cbgraph}
    
    This theorem is formalized as \codelink{VC.connected\_cbgraph\_with\_failure}{https://github.com/dz7903/cbgraphs-formalization/blob/322519adb499fbf4ec8516ba46523e710cf8b897/FormalCbgraphs/Soundness.lean\#L129} in Lean.
\end{proof}

Finally, we formally state Theorem~\ref{thm:correctness} and its fault-tolerance version. The proofs are direct application of Theorem~\ref{thm:connected-cbgraph} and $\vcprop$.
\begin{manualtheorem}{\ref{thm:correctness}}
    Given verification conditions hold, for any fair schedule $S$ and any node $v\in V$, $$\exists\tau.\ \forall t\ge\tau.\ \sigma_S(v,t)\in Y(v).$$

    This theorem is formalized as \codelink{VC.correctness}{https://github.com/dz7903/cbgraphs-formalization/blob/322519adb499fbf4ec8516ba46523e710cf8b897/FormalCbgraphs/Soundness.lean\#L146} in Lean.
\end{manualtheorem}
\begin{manualtheorem}{\ref{thm:correctness}, fault-tolerance version}
    Given verification conditions hold, for any schedule $S$ that is fair with at most $k$ failures, if the \cbgraph $CB$ is $k$-fail-connected, then for any node $v\in V$, $$\exists\tau.\ \forall t\ge\tau.\ \sigma_S(v,t)\in Y(v).$$

    This theorem is formalized as \codelink{VC.correctness\_with\_failure}{https://github.com/dz7903/cbgraphs-formalization/blob/322519adb499fbf4ec8516ba46523e710cf8b897/FormalCbgraphs/Soundness.lean\#L155} in Lean.
\end{manualtheorem}

\subsubsection{Completeness theorem}

In \S\ref{subsec:alg1} we mentioned that \sysname is complete if every transfer function in the network is \emph{strict monotone} and $<$ is a well-order on $R$. We have also formalized this in Lean as shown below.

\begin{theorem}[Completeness of \sysname]
    Assume $<$ is a well-order on $R$. If all transfer functions $f_e$ in $N$ are strict monotone, \IE $x<y$ implies $f(x)<f(y)$, and property $Y$ holds eventually stably for $\sigma_S(v,t)$ for each $v$ and any fair schedule $S$, then there exist \interfaces $I,Q$ and a connected \cbgraph such that all verification conditions in Fig.~\ref{fig:VCs} hold.
\end{theorem}
\begin{proof}[Proof sketch]
    For each path $p=[v_1,v_2,\cdots,v_n]$ in the network $N$, define the route for $p$ as $$s_p=f_{(v_{n-1},v_n)}(\cdots f_{(v_1,v_2)}(\mathrm{init}(v_1))).$$
    
    Let $\pi(v)$ be all paths in $N$ that end with $v$. We then define
    \begin{align*}
        I(v)&=\bigcup_{p\in\pi(v)}\{s_p\} \\
        Q(v)&=\left\{\min_{p\in\pi(v)}s_p\right\} \\
        CBRoots&=\{v\in V\mid[v]\in Q(v)\} \\
        CBEdges&=\{(u,v)\in E\mid\exists p\in\pi(u).p+\!\!\!\!+[v]\in Q(v)\}
    \end{align*}
    (the existence of $\min_{p\in\pi(v)}s_p$ comes from the well-order; $p+\!\!\!\!+[v]$ is to append $v$ at the end of path $p$), all VCs in Fig.~\ref{fig:VCs} except $\vcprop$ and $\vcedge$ can be easily verified.

    The tricky ones are $\vcedge$ and the connectedness of \cbgraph. In fact, thanks to strict monotonicity, we can show the following key property: if $s_{[p,v]}\in Q(v)$ for $p\in\pi(u)$ and $(u,v)\in E$, then $s_p\in Q(u)$. $\vcedge$ and connectedness come from this property.

    For $\vcprop$, note that Theorem~\ref{thm:connected-cbgraph} ensures each node $v$ abstractly converges to $Q(v)$ in any fair schedule; by definition of $Q$, it concretely converges to $\min_{p\in\pi(v)}$. Since $Y(v)$ holds eventually stably, it must belong to $Y(v)$, \IE $\vcprop$ verified.

    This theorem is formalized as \codelink{completeness}{https://github.com/dz7903/cbgraphs-formalization/blob/322519adb499fbf4ec8516ba46523e710cf8b897/FormalCbgraphs/Completeness.lean\#L45} in Lean.
\end{proof}

\subsubsection{The maximality of generated \cbgraph}

In \S\ref{subsec:cbgraphs} we claimed that \sysname automatically \cbgenerates a \emph{maximal} \cbgraph. We formalize and prove its maximality as the following theorem:

\begin{theorem}[Maximality of generated \cbgraph]
    Given a network instance $N$, its property $Y$ and annotations $I,Q$, if $(\mathrm{CBRoots},\mathrm{CBEdges})$ is the \cbgraph generated by Algorithm~\ref{alg:verification}, then for any \cbgraph $(\mathrm{CBRoots}',\mathrm{CBEdges}')$ that satisfies the verification conditions in Fig.~\ref{fig:VCs}, we have $\mathrm{CBRoots}'\subseteq\mathrm{CBRoots}$ and $\mathrm{CBEdges}'\subseteq\mathrm{CBEdges}$.
\end{theorem}
\begin{proof}
    Proof is trivial. This theorem is formalized as \codelink{generatedCBRoots\_maximal}{https://github.com/dz7903/cbgraphs-formalization/blob/322519adb499fbf4ec8516ba46523e710cf8b897/FormalCbgraphs/Maximality.lean\#L17} and \codelink{generatedCBEdges\_maximal}{https://github.com/dz7903/cbgraphs-formalization/blob/322519adb499fbf4ec8516ba46523e710cf8b897/FormalCbgraphs/Maximality.lean\#L22}.
\end{proof}

\subsection{Encoding BGP protocol in SMT}

\label{appendix-subsec:smt}

Here we present a way to encode BGP protocol in SMT theories.
A BGP route is encoded in \sysname as a quadruple $S=(\mathrm{prefix}, \mathrm{lp}, \mathrm{aspath}, \mathrm{comm})$. (The actual BGP protocol includes more fields. For simplicity, we do not show the model of all BGP features here.) Here $\mathrm{prefix}$ is a 32-bit vector (assuming IPv4) indicating the destination prefix, $\mathrm{lp}$ is an integer indicating local preference, $\mathrm{aspath}$ is an ordered sequence of AS numbers for each router this message has traversed (the AS-path in BGP), and $\mathrm{comm}$ is a finite set of community tags. We use expression $\mathrm{prefix}(s)$ to denote the prefix field of a route $s$.

In BGP, the merge function encodes the preference order between two routes, checking the local preference first (higher is preferred), and if it is the same, then the length of the AS-path (shorter path is preferred). Accordingly, the merge operator is modeled as:
$$\mathrm{merge}(s_1,s_2)=\begin{cases}
    s_1 & \mathrm{lp}(s_1)>\mathrm{lp}(s_2) \\
    s_2 & \mathrm{lp}(s_1)<\mathrm{lp}(s_2) \\
    s_1 & \mathrm{lp}(s_1)=\mathrm{lp}(s_2) \\
    &\land\ \mathrm{len}(\mathrm{aspath}(s_1))<\mathrm{len}(\mathrm{aspath}(s_2)) \\
    s_2 & \mathrm{lp}(s_1)=\mathrm{lp}(s_2) \\
    &\land\ \mathrm{len}(\mathrm{aspath}(s_1))>\mathrm{len}(\mathrm{aspath}(s_2)) \\
    s_1 & \text{otherwise}
\end{cases}$$
Cases can be expressed through if-then-else expressions, $\mathrm{len}(l)$ is a pre-defined function in the theory of sequences, returning the length of the list $l$. (We write $\mathrm{len}(s)$ in short for $\mathrm{len}(\mathrm{path}(s))$ in the paper.) For the last case, two routes are equally preferred, and one of them is chosen arbitrarily.

The transfer functions are modeled based on the configuration and its semantics. For example, a \texttt{set community 100:2} statement in Cisco’s configuration language sets the communities field to be $\{100:2\}$, which can be modeled as:
\begin{align*}
\mathrm{tr}(s,s')\Leftrightarrow&\ \mathrm{comm}(s')=\{100:2\} \\
&\land\mathrm{prefix}(s')=\mathrm{prefix}(s) \\
&\land\mathrm{lp}(s')=\mathrm{lp}(s) \\
&\land\mathrm{aspath}(s')=\mathrm{aspath}(s)
\end{align*}
\subsection{Error Debugging Process}
\label{appendix-subsec:debug}

\begin{figure*}
    \centering
    \begin{tabular}{cll}
        \toprule
        VC and Counterexample & User Beliefs & Cause and Suggested Action \\
        \midrule
        \multirowcell{2.5}{
            $VC_{\mathrm{Init}}$ is violated: \\
            $s_v=\mathrm{init}(v)\notin I(v)$}
        & case 1: $s_v$ \xmark & $\mathrm{init}(v)$ is wrong, \textbf{bug in configuration} \\
        \cmidrule{2-3}
        & case 2: $s_v$ \cmark & $I(v)$ is too strong, weaken it \\
        \midrule
        \multirowcell{4}{
            $VC_{\mathrm{Inv}}$ is violated: \\
            $\begin{aligned}
                &s_u\in I(u), s_v\in I(v), \\
                &s_v'=s_v\oplus f_{(u,v)}(s_u) \notin I(v)
            \end{aligned}$}
        & case 1: $s_u$ or $s_v$ \xmark & $I(u)$ or $I(v)$ is too weak, strengthen it \\
        \cmidrule{2-3}
        & case 2: $s_u,s_v$ and $s_v'$ \cmark & $I(v)$ is too strong, weaken $I(v)$ \\
        \cmidrule{2-3}
        & case 3: $s_v'$ \xmark & $f_{(u,v)}$ is wrong, 
        \textbf{bug in configuration} \\
        \midrule 
        \multirowcell{2.5}{
            $VC_{\mathrm{Prop}}$ is violated: \\
            $s_v\in Q(v), s_v \notin Y(v)$}
        & case 1: $s_v$ \xmark & $Q(v)$ is too weak, strengthen it \\
        \cmidrule{2-3}
        & case 2: $s_v$ \cmark & \textbf{property $Y(v)$ fails on a plausible route} \\
        \midrule \midrule
        \multirowcell{5}{
            $VC_{\mathrm{CBroot}}$ is violated \\
             (first conjunct): \\
            $s_v=\mathrm{init}(v)\notin Q(v)$}
        & \multirowcell{2}[0pt][l]{case 1: $v$ should not \\ be a \cbroot} & \multirowcell{2}[0pt][l]{skip this counterexample} \\
        \\
        \cmidrule{2-3}
        & case 2: $s_v$ \xmark & $\mathrm{init}(v)$ is wrong, \textbf{bug in configuration} \\
        \cmidrule{2-3}
        & case 3: $s_v$ \cmark & $Q(v)$ is too strong, weaken it \\
        \midrule
        \multirowcell{6}{
            $VC_{\mathrm{CBroot}}$ is violated \\
            (second conjunct): \\
            $\begin{aligned}
                &s_u\in I(u), s_v\in Q(v), \\
                &s_v'=s_v\oplus f_{(u,v)}(s_u) \notin Q(v)
            \end{aligned}$}
        & \multirowcell{2}[0pt][l]{case 1: $v$ should not \\ be a \cbroot} & \multirowcell{2}[0pt][l]{skip this counterexample} \\
        \\
        \cmidrule{2-3}
        & case 2: $s_u$ or $s_v$ \xmark & $I(u)$ or $Q(v)$ is too weak, strengthen it \\
        \cmidrule{2-3}
        & case 3: $s_u,s_v$ and $s_v'$ \cmark & $Q(v)$ is too strong, weaken it \\
        \cmidrule{2-3}
        & case 4: $s_v'$ \xmark & $f_{(u,v)}$ is wrong, 
        \textbf{bug in configuration} \\
        \midrule
        \multirowcell{6}{
            $VC_{\mathrm{CBedge}}$ is violated: \\
            $\begin{aligned}
                &s_u\in Q(u),s_v\in I(v), \\
                &s_v'=s_v\oplus f_{(u,v)}(s_u) \notin Q(v)
            \end{aligned}$}
        & \multirowcell{2}[0pt][l]{case 1: $(u,v)$ should \\ not be a \cbedge} & \multirowcell{2}[0pt][l]{skip this counterexample} \\
        \\
        \cmidrule{2-3}
        & case 2: $s_u$ or $s_v$ \xmark & $Q(u)$ or $I(v)$ is too weak, strengthen it \\
        \cmidrule{2-3}
        & case 3: $s_u,s_v$ and $s_v'$ \cmark & $Q(v)$ is too strong, weaken it \\
        \cmidrule{2-3}
        & case 4: $s_v'$ \xmark & $f_{(u,v)}$ is wrong, 
        \textbf{bug in configuration} \\
        \bottomrule
    \end{tabular}
    \caption{Error Debugging Process: Counterexample shows the VC violation in the first column,  ``$s_v$ \cmark'' in the second column indicates the user's belief that route $s_v$ is a \emph{plausible} route that be selected by $v$, and ``$s_v$ \xmark'' indicates the user's belief that $s_v$ is an \emph{implausible} route that cannot be selected by $v$; the last column shows the cause and suggested action for each case.}
    \label{fig:errors-full}
\end{figure*}

Fig.~\ref{fig:errors-full} presents how a user of \sysname can debug the VC violations found during Algorithm~\ref{alg:verification}. The top part (separated by a double-line from the bottom part) describes violations in the essential VCs ($\vcinit$, $\vcinv$ or $\vcprop$) that result in immediate failure in phase 1; the bottom part describes VC violations (in $\vcroot, \vcedge$) that are saved in phase 1, to possibly recover from a failure in phase 2 if the connectedness check on the constructed \cbgraph fails. We now explain some cases in Fig.~\ref{fig:errors-full} with more details below.

\paragraph*{Violation of $\vcinv$.} To illustrate the process in detail, consider a violation of $\vcinv$ (shown in Fig.~\ref{fig:errors-full}). 
Recall that $\vcinv$ (shown earlier in Fig.~\ref{fig:VCs}) states: $\forall s_u,s_v.\ s_u\in I(v)\land s_v\in I(v)\Rightarrow s_v\oplus f_{(u,v)}(s_u)\in I(v)$. 
When this VC check fails, a counterexample is an assignment to $s_u$ and $s_v$ such that the formula is false. That is, $s_u\in I(u)$ and $s_v\in I(v)$ are true, but $s_v\oplus f_{(u,v)}(s_u)\in I(v)$ is false (shown in column~1,  Fig.~\ref{fig:errors-full}).

To diagnose the cause of this VC violation, the user performs a case analysis based on their beliefs (cases shown in column~2).
They first decide whether $s_u$ and $s_v$ are \emph{plausible} routes that could be selected at nodes $u$ and $v$, respectively, based on their beliefs about the routing behavior at $u$ and $v$. Plausibility indicates whether the user believes that such $s_u$ and $s_v$ can occur in a real simulation; in the case of $\vcinv$ this means there exists a schedule $S$ such that $s_u$ and $s_v$ are selected by $u$ and $v$ at some time.

The first case (case 1 in the second column) is if one or both of $s_u, s_v$ are not plausible.
In this case, the corresponding $I(u)$ or $I(v)$ is considered erroneous, \IE it is too weak. 
To fix this, the \interface $I(u)$ or $I(v)$ should be \emph{strengthened}, \IE it should be made more restrictive by excluding $s_u$ and/or $s_v$.

Otherwise, if both $s_u$ and $s_v$ are plausible, the user will next examine $s'_v = s_v \oplus f_{(u,v)(s_u)}$, the newly selected route at $v$. 
Note that $s'_v$ does not belong to $I(v)$ according to the counterexample. 
If $s'_v$ is plausible (case 2), then $I(v)$ is too strong and should be weakened to include $s'_v$. 
However, if $s'_v$ is not plausible (case 3), then the transfer function $f_{(u,v)}$ is considered erroneous and should be repaired.

\paragraph*{Violation of $\vcprop$.}
Now, consider a failure where $\vcprop$ is violated. The first case is that the counterexample route $s_v$ is not plausible, \IE $s_v$ is a \emph{spurious} counterexample. In this case, the user can strengthen $Q(v)$ to exclude $s_v$ during refinement. This is similar to classic counterexample guided abstraction refinement (CEGAR)~\cite{ClarkeCegar00}. However, if $s_v$ is plausible, then \sysname has found a bug for property $Y(v)$. The user could either accept the reported bug, or weaken the property to $Y'(v)$ and (incrementally) verify $\vcprop$ again. 

\paragraph*{Violations of $\vcroot$ and $\vcedge$.}
\sysname can fail in phase~2 if there does not exist any connected \cbgraph according to $\vcroot$ and $\vcedge$ (identified in phase~1), \IE if some nodes in the network are found to be \emph{unconnected} to any \cbroot via \cbedges.
For an unconnected node $v$, $\vcroot$ is invalid at $v$, and $\vcedge$ is invalid for every incoming edge to $v$. 
However, unlike in phase 1 failures, not every invalid VC requires an action for repair:
only that \emph{at least one} of the VCs is repaired (to connect the node).

Users rarely need to repair $\vcroot$; in practice \cbroots are expected to be the destination itself or border routers with direct announcements to external destinations.
For repairing $\vcedge$, the user can examine the node's incoming edges and diagnose why $\vcedge$ fails for all of them.
Repairing all the failures may not be necessary to connect the node but can provide better fault tolerance.
The repair process for fixing $\vcedge$ is similar to that for $\vcinv$.

\subsection{Error Debugging Process: an Example}

\label{appendix-subsec:debug-example}

\begin{figure}
    \centering
    \includegraphics[width=\linewidth]{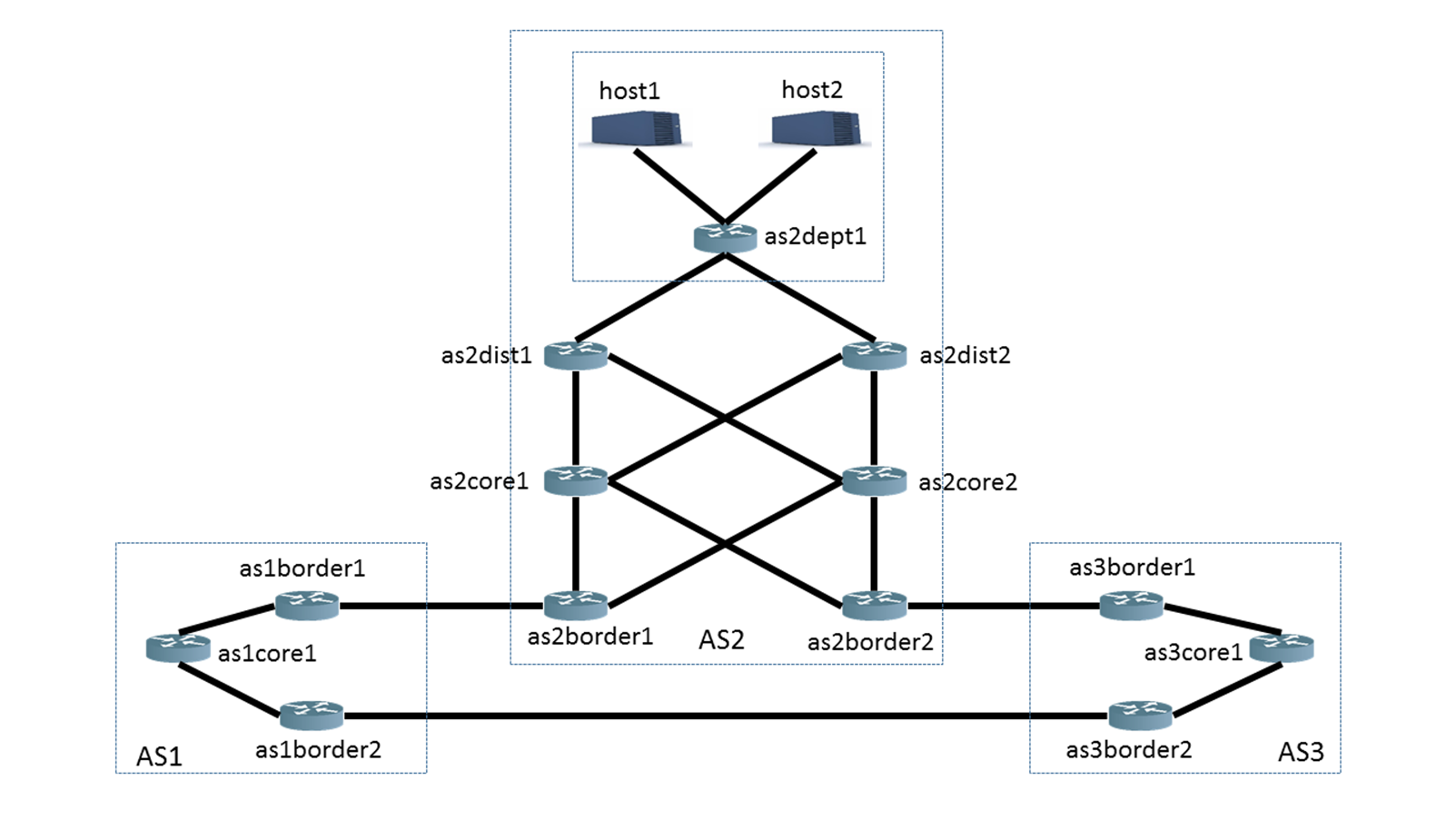}
    \caption{Topology of a campus network from the Batfish tutorial \cite{batfish-tutorials-bgp}.
    }
    \label{fig:example-network-topology}
\end{figure}

\begin{figure*}[ht]
    \centering
    \begin{tabular}{ccccc}
        \toprule
        Step & \Interface & Failed check and location & Cause & Action \\
        \midrule
        \multirowcell{3}{1} & \multirowcell{3}{$\begin{aligned}
            Q_{AS2}&=\{s\mid s\ne\infty\}, \\
            I_{AS2}&=R
        \end{aligned}$} & \multirowcell{3}{\textbf{\convcheck} at edges \\ (as2dist1, as2dept1) and \\ (as2dist2, as2dept1)} & \multirowcell{3}{as2dept1 will drop \\
        route without \\ community 3:2} & \multirowcell{3}{Refine $Q_{AS2}$ \\ and $I_{AS2}$} \\
        \\
        \\
        \midrule
        \multirowcell{6}{2} & \multirowcell{6}{$\begin{aligned}
            Q_{AS2}&=\left\{s\ \middle|\ \begin{aligned}
                &s\ne\infty \\
                &\land 3:2\in\mathrm{tag}(s)
            \end{aligned}\right\} \\
            I_{AS2}&=\left\{s\ \middle|\ \begin{aligned}
                &s=\infty \\
                &\lor 3:2\in\mathrm{tag}(s)
            \end{aligned}\right\}
        \end{aligned}$} & \multirowcell{6}{\textbf{\Invcheck} at edge \\ (as1border2, as2border1)} & \multirowcell{6}{as2border1 can select \\ a route without \\ community 3:2 \\ (but with 1:2)} & \multirowcell{6}{Refine $Q_{AS2}$ \\ and $I_{AS2}$} \\
        \\
        \\
        \\
        \\
        \\
        \midrule
        \multirowcell{8}{3} & \multirowcell{8}{$\begin{aligned}
            Q_{AS2}&=\left\{s\ \middle|\ \begin{aligned}
                &s\ne\infty \\
                &\land\left(\begin{aligned}
                    & 1:2\in\mathrm{tag}(s)\\
                    & \lor 3:2\in\mathrm{tag}(s)
                \end{aligned}\right)
            \end{aligned}\right\} \\
            I_{AS2}&=\left\{s\ \middle|\ \begin{aligned}
                &s=\infty \\
                &\lor\left(\begin{aligned}
                    & 1:2\in\mathrm{tag}(s)\\
                    & \lor 3:2\in\mathrm{tag}(s)
                \end{aligned}\right)
            \end{aligned}\right\}
        \end{aligned}$} & \multirowcell{8}{\textbf{\convcheck} at edges \\ (as2dist1, as2dept1) and \\ (as2dist2, as2dept1) \\} & \multirowcell{8}{as2dept1 wrongly \\ dropped routes \\ with
        community 1:2} & \multirowcell{8}{Repair \\ $f_{(\textrm{as2dist1},\textrm{as2dept1})}$ and \\ $f_{(\textrm{as2dist2},\textrm{as2dept1})}$} \\
        \\
        \\
        \\
        \\
        \\
        \\
        \\
        \midrule
        4 & same as above & \multicolumn{3}{c}{all checks pass!} \\
        \bottomrule
    \end{tabular}
    \caption{Error debugging process: an example (for the network in Fig.~\ref{fig:example-network-topology}).}
    \label{fig:example-debug}
\end{figure*}

Fig.~\ref{fig:example-network-topology} shows the topology of a campus network taken from Batfish tutorial~\cite{batfish-tutorials-bgp}. Fig.~\ref{fig:example-debug} shows an example of the debugging process based on the network in Fig.~\ref{fig:example-network-topology}. Suppose the network operator modifies the configuration by adding communities 1:2 and 3:2 when routes are imported to AS2 from AS1 and AS3, respectively, and only allows routes with such communities to be imported to as2dept1. However, the operator 
makes a mistake in the configuration at as2dept1, such that it only permits routes with community 3:2.

The operator wishes to verify that any routes that originate from AS2 can reach as2dept1. The operator first specifies $Q_{AS2}$ as the set of any reachable route, \IE $\{s\mid s\ne\infty\}$ and $I_{AS2}$ as the set of all routes, \IE $R$ (shown as step~1 in Fig.~\ref{fig:example-debug}). This results in a disconnected \cbgraph, and the counterexamples of \convcheck at the edges (as2dist1,as2dept1) and (as2dist2-as2dept1) show that routes without community 3:2 are dropped at these edges.

Thus, the operator can refine $Q_{AS2}$ and $I_{AS2}$ to \emph{include} routes with community 3:2 (shown as step 2 in Fig.~\ref{fig:example-debug}). However, this time, \invcheck fails; the counterexample at the edge (as1border2, as2border1) shows that a route with community 1:2 but without community 3:2 is imported at AS2, thus violating the invariant $I_{AS2}$.

Because the operator believes that the route with community 1:2 is \emph{plausible} in AS2, they can then refine $Q_{AS2}$ and $I_{AS2}$ to include both communities 1:2 and 3:2 (step 3 in Fig.~\ref{fig:example-debug}). This results again in a disconnected \cbgraph. The counterexamples of \convcheck show that routes with community 1:2 but without community 3:2 are dropped.

This time, since routes with community 1:2 is \emph{plausible} in AS2, the operator knows that the configurations are wrong. Therefore, the operator repairs the configuration at both edges (as2dist1, as2dept1) and (as2dist2, as2dept2) either manually or automatically, and when suitably repaired, all checks will pass after the repair.
\subsection{A counterexample on the completeness of \sysname algorithm}

\label{appendix-subsec:completeness}

Now we present a counterexample showing that \sysname is \emph{incomplete}, \IE there exists a network instance $N$ and its properties $Y$ such that $Y$ are eventually-stable,
but there are no \interfaces $Q,I$ such that \sysname could \cbgenerate a connected \cbgraph. This counterexample has been formalized in Lean.

Our example is again inspired by Daggit and Griffin's work~\cite{daggitt2018rate}. The network is shown in Fig.~\ref{fig:incompleteness-network}, with four nodes $A,B,C,D$ and four directed edges $AB,AC,BC,CD$.

\begin{figure}[t]
    \centering
    \begin{tikzpicture}
        \node[node1](A){A};
        \node[node2,above right=0.7cm and 0.4cm of A](B){B};
        \node[node2,right=1cm of A](C){C};
        \node[node2,right=1cm of C](D){D};
        \draw[->](A)--(B);
        \draw[->](A)--(C);
        \draw[->](B)--(C);
        \draw[->](C)--(D);
    \end{tikzpicture}
    \caption{Example network for the incompleteness of \sysname. $A$ is the destination.}
    \label{fig:incompleteness-network}
\end{figure}

\begin{figure}[t]
    \centering
    \vspace{0.1in}
    \begin{tabular}{cccccccc}
        \toprule
        $r=$ & $r_0$ & $r_1$ & $r_2$ & $r_3$ & $r_4$ & $\infty$ \\
        \midrule
        $f_{AB}(r)=f_{BC}(r)=$ & $r_1$ & $r_2$ & $r_3$ & $r_4$ & $\infty$ & $\infty$ \\
        \midrule
        $f_{AC}(r)=$ & $r_3$ & $r_4$ & $\infty$ & $\infty$ & $\infty$ & $\infty$ \\
        \midrule
        $f_{CD}(r)=$ & $r_1$ & $r_2$ & $\infty$ & $r_4$ & $\infty$ & $\infty$ \\
        \bottomrule
    \end{tabular}
    \caption{Transfer functions of the network in Fig.~\ref{fig:incompleteness-network}. Each row represents a group of transfer functions, and each column shows the value of $f_e(r)$ for each different $r$.}
    \label{fig:incompleteness-transfer}
\end{figure}

Let the set of routes be $R=\{r_0,r_1,r_2,r_3,r_4,\infty\}$ (where $r_0<r_1<r_2<r_3<r_4<\infty$, and $\oplus$ chooses the smaller one). Let the initial routes of the network be: 
$$\mathrm{init}(v)=\begin{cases}
    r_0 & v=A \\
    \infty & \text{otherwise}
\end{cases}$$ and the transfer functions of the network be shown in Fig.~\ref{fig:incompleteness-transfer}. In Fig.~\ref{fig:incompleteness-transfer}, $f_{AB},f_{BC}$ increase the route by one level, $f_{CD}$ increases the route by one level except that $r_2$ is dropped, and $f_{AC}$ increases the route by three levels.
Let the properties be: 
$$Y(v)=\begin{cases}
    \{s\mid\mathrm{True}\} & v\in\{A,B,C\} \\
    \{\infty\} & v=D
\end{cases}$$
We now state the following claims.

\begin{fact}
$Y$ properties are eventually-stable in the network in Fig.~\ref{fig:incompleteness-network}. However, there are no \interfaces $Q,I$ such that the \cbgraph \cbgenerated  according to Fig.~\ref{fig:VCs} is connected.
\end{fact}
\begin{proof}
First we need to show that $Y$ properties are eventually-stable in the network. In fact, we can show that in any fair schedule $S$, $\sigma_S(v,t)$ concretely converges to $r_0,r_1,r_2,\infty$ for $v=A,B,C,D$. However the argument is quite tedious; we refer to the formalized proofs in Lean (Appendix~\ref{appendix-subsec:repo}).

To show that no \interfaces $Q,I$ with connected \cbgraph \cbgenerated by Fig.~\ref{fig:VCs} exist, suppose there are such $Q$ and $I$. We have the following facts:
\begin{enumerate}
    \item $r_2\in Q(C)$. This is because Theorem~\ref{thm:connected-cbgraph} ensures $\sigma_S(C,t)$ must abstractly converge to $Q(C)$ in any fair schedule $S$, and we already know that $\sigma_S(C,t)$ concretely converges to $r_2$ in such $S$.
    \item $r_4\in I(d)$. This is because Lemma~\ref{lemma:invariance} ensures $\sigma_S(D,t)\in I(v)$ for any $t$ in any schedule $S$, and we can construct a concrete $S$ (specifically the synchronous schedule) such that $\sigma_S(D,2)=r_4$.
    \item $r_4\notin Q(D)$. This is because $\vcprop$ requires $Q(D)\subseteq Y(D)=\{\infty\}$.
    \item $(C,D)$ cannot be a \cbedge. This is because $\vcedge$ is not valid for $(C,D)$ --- we have $r_2\in Q(C)$, $r_4\in I(D)$ but $f_{CD}(r_2)\oplus r_4=r_4\notin Q(D)$.
\end{enumerate}

Note that $D$ also cannot be a \cbroot, otherwise Lemma~\ref{lemma:cbroots} ensures $\sigma_S(D,t)=\infty$ for any $t$, but we can construct a concrete $S$ such that $\sigma_S(D,2)=r_4$. Therefore, no path in \cbgraph can end with $D$, which contradicts with the assumption that the \cbgraph is connected.

This fact is formalized as \codelink{Counterexample.incomplete}{https://github.com/dz7903/cbgraphs-formalization/blob/322519adb499fbf4ec8516ba46523e710cf8b897/FormalCbgraphs/Incompleteness.lean\#L175} in Lean.
\end{proof}

\subsection{Detailed Specifications for Benchmarks}

\label{appendix-subsec:benchmark}

We describe details about the formal specifications on synthetic fattree and Internet2 benchmarks in this section.

\paragraph*{Fattrees networks.} We generate fattrees parameterized by the number of pods, $k$. The number of nodes is $1.25k^2$, and the number of edges is $k^3$.
Nodes of fattrees are divided into edge-nodes (top-of-the-rack nodes), aggregation-nodes and core nodes, and edge-nodes and aggregation-nodes must belong to certain pods \cite{al2008scalable}. For \bench{Hijack}, we add an extra "hijacker" node connecting with all the core-nodes.

Our network model requires us to choose a destination when we verify any property. We always use \texttt{edge0\_0}, the first edge-node in the first pod, as the destination. In the formulas below, we refer this destination as an constant $d$. The destination has an initial route with zero path length and an internal destination prefix. In \bench{Hijack}, we further give the hijacker-node an arbitrary initial route with an internal destination, to "hijack" the internal routes.
Also, our fattree network uses eBGP only; every node is assigned a different AS number.

\paragraph*{Reachability.} In \bench{Reachability} benchmark the policies between routers are simply forward any routes (that is, a single $\mathbf{permit}$ action). The properties to be verified are:
$$Y(v)=\{s\mid s\ne\infty\}.$$

We specify \interfaces as follows:
\begin{align*}
    I(v)&=\{s\mid s=\infty\lor s\ne\infty\} \\
    Q(v)&=\{s\mid s\ne\infty\}
\end{align*}

\paragraph*{Path length.} In \bench{PathLength} benchmark the policies are the same as \bench{Reachability}. We want to verify a stronger property that each router will select the shortest path to destination finally:
$$Y(v)=\{s\mid s\ne\infty\land\mathrm{len}(s)=dist(v)\}$$
where $dist(v)$ is the distance between node $v$ and the destination, and $\mathrm{len}(s)$ is the path length field in BGP route $s$.

To prove this property, we need to specify that:
\begin{enumerate}
    \item The local preference of routes are the same as the default value (100), so that longer paths do not have higher local preference to be selected;
    \item Before convergence, the router may select another route, but the length must not be less than the distance $dist(v)$.
\end{enumerate}

Therefore, the \interfaces for \bench{PathLength} are
\begin{align*}
I(v)&=\{s\mid s=\infty\Rightarrow\mathrm{lp}(s)=100\land\mathrm{len}(s)\ge dist(v)\} \\
Q(v)&=\{s\mid s\ne\infty\land\mathrm{lp}(s)=100\land\mathrm{len}(s)=dist(v)\}
\end{align*}

\paragraph*{Valley-free.} In \bench{ValleyFree} we check the valley-free property in fattrees, \IE no valley path (up-down-up path) is selected. To ensure this property, we modify the policies so that:
\begin{itemize}
    \item If a route is advertised through a ``down'' link (core to aggregation or aggregation to edge), a community tag 1:0 will be added.
    \item Routes with community tag 1:0 will be dropped along ``up'' links (aggregation to core or edge to aggregation).
\end{itemize}

We verify that each node will be reachable (reachability), and nodes that connect with destination with only ``up'' links (we call such nodes the \emph{uphill} nodes, and write $\mathrm{uphill}(d)$ as the set of all the uphill nodes when destination is $d$)
will not select a route with community tag 1:0.
$$Y(v)=\{s\mid s\ne\infty\land(v\in\mathrm{uphill}(d)\Rightarrow 1:0\notin\mathrm{comm}(s))\}$$
where $\mathrm{comm}(s)$ is the field of community tags of the BGP route $s$.

Due to asynchrony, a route carrying community $1:0$ may be selected before convergence, even if it is currently a shortest route at an uphill node. At convergence, however, every node selects a shortest route, and uphill nodes must not select a route carrying $1:0$.
The \interfaces should be specified as
\begin{align*}
I(v)&=\{s\mid s\neq\infty\Rightarrow
(lp(s)=100\land len(s)\ge dist(v))\} \\
Q(v)&=\left\{s\ \middle|\ \begin{aligned}
&s\ne\infty\land\mathrm{lp}(s)=100\land\mathrm{len}(s)=dist(v) \\
&\land(v\in\mathrm{uphill}(d)\Rightarrow 1:0\notin\mathrm{comm}(s))
\end{aligned}\right\}
\end{align*}

\paragraph*{Hijack.} In \bench{Hijack} we add a hijacker connecting with all core nodes as an external node that can advertise any route from outside.

The policies within the fattree are the same as those in \bench{Reachability} property (a single $\mathbf{permit}$ action). However, the import policies at core nodes from the hijacker will drop the route if it has an internal prefix. Also, the hijacker-node does not accept any route from the fattree (by dropping routes from core nodes).

We want to verify that no route coming from hijacker should be selected (route filtering property). To check whether a route is from hijacker or not, we utilize the AS path field in BGP --- for a BGP route $s$, we use $\mathrm{ASPath}(s)$ as the set of AS numbers in its AS path field. We also use $\mathrm{IntPrefixes}$ as the set of internal prefixes and $\mathrm{HijackAS}$ as the AS number of hijacker.

We specify the properties as
$$Y(v)=\begin{cases}
\{s\mid s\ne\infty\land\mathrm{HijackAS}\notin\mathrm{ASPath}(s)\} & v\ne\mathrm{hijacker} \\
\{s\mid \mathrm{True}\} & v=\mathrm{hijacker}
\end{cases}$$

The \interfaces are specified as follows (note that $I(\mathrm{hijacker})=Q(\mathrm{hijacker})$:
\begin{align*}
I(v)&=\begin{cases}
    \{s\mid s=\infty\lor\mathrm{HijackAS}\notin\mathrm{ASPath}(s)\} & v\ne\mathrm{hijacker} \\
    \{s\mid s\ne\infty\land\mathrm{prefix}(s)\in\mathrm{IntPrefixes}\} & v=\mathrm{hijacker}
\end{cases} \\
Q(v)&=\begin{cases}
    \{s\mid s\ne\infty\land\mathrm{HijackAS}\notin\mathrm{ASPath}(s)\} & v\ne\mathrm{hijacker} \\
    \{s\mid s\ne\infty\land\mathrm{prefix}(s)\in\mathrm{IntPrefixes}\} & v=\mathrm{hijacker}
\end{cases}
\end{align*}

\paragraph*{Internet2 network.} The Internet2 network has 10 internal nodes (sometimes we also refer as core nodes) and 283 external nodes. We consider multiple destination; in \bench{BlockToExternal} and \bench{InternalReachability} the destination could be any internal node, while in \bench{NoMartians} the destination could be any external node. In either case, we use $d$ as a symbolic variable representing the destination. The \cbgraph will be generated and checked on connectedness for each concrete destination.

\paragraph*{Block-to-external.} In \bench{BlockToExternal} property we want to verify that internal routes with a BTE community will not be advertised to the external peers.
In configuration, the destination has an initial route, which could contain BTE community or not.
We verify that for external neighbors ($v\in\mathrm{ExternalNodes}$), no route is selected if the initial route of the destination contains a BTE community:
\begin{align*}
    Y(v)=\begin{cases}
        \{s\mid s=\infty\} & v\in\mathrm{ExternalNodes}\land\mathrm{BTE}\in\mathrm{comm}(\mathrm{init}(d)) \\
        \{s\mid \mathrm{True}\} & \text{otherwise}
    \end{cases}
\end{align*}

The \interfaces are specified as
$$I(v)=Q(v)=\begin{cases}
\{s\mid s\ne\infty\Rightarrow\mathrm{BTE}\in\mathrm{comm}(s)\} \\
\qquad\qquad v\in\mathrm{InternalNodes}\land\mathrm{BTE}\in\mathrm{comm}(\mathrm{init}(d)) \\
\{s\mid s=\infty\} \\
\qquad\qquad v\in\mathrm{ExternalNodes}\land\mathrm{BTE}\in\mathrm{comm}(\mathrm{init}(d)) \\
\{s\mid\mathrm{True}\} \\
\qquad\qquad\mathrm{BTE}\notin\mathrm{comm}(\mathrm{init}(d))
\end{cases}$$

\paragraph*{No-martians.} In \bench{NoMartians}, we verify that no routes with a martian prefix (which is a list of certain IP prefixes) can be selected by an internal node. We give the destination (which is an external node) arbitrary initial route. We then verify that
$$Y(v)=\begin{cases}
    \{s\mid s\ne\infty\Rightarrow\mathrm{prefix}(s)\notin\mathrm{MartianPrefix}\} & v\in\mathrm{InternalNode} \\
    \{s\mid \mathrm{True}\} & v\in\mathrm{ExternalNode}
\end{cases}$$

The \interfaces are easy. They are specified the same as the property itself.
$$I(v)=Q(v)=Y(v)$$

\paragraph*{Internal reachability.} In \bench{InternalReachability}, we verify that any route advertised by an internal node will be selected by other internal nodes (actually, Internet2's policies between internal nodes do not drop or modify routes). We still give an arbitrary initial route for the destination. We verify that
$$Y(v)=\begin{cases}
    \{s\mid s\ne\infty\} & v\in\mathrm{InternalNode} \\
    \{s\mid \mathrm{True}\} & v\in\mathrm{ExternalNode}
\end{cases}$$
with \interfaces
\begin{align*}
    I(v)&=\{s\mid\mathrm{True}\} \\
    Q(v)&=Y(v)
\end{align*}

\end{document}